\documentclass[11pt,a4paper]{article} %
\usepackage[left=1in,right=1in,top=1in,bottom=1in]{geometry} %
\usepackage{amsmath, amssymb, bbm, amsthm, thmtools, graphicx, etoolbox, subcaption, stmaryrd, xcolor}

\usepackage[colorlinks=true,allcolors=blue]{hyperref}
\usepackage[capitalize]{cleveref}

\usepackage{xcolor}

\Crefname{algocf}{Algorithm}{Algorithms}

\usepackage[algo2e, ruled, vlined]{algorithm2e}
\usepackage{xcolor}

\SetCommentSty{mycommfont}

\usepackage{fullpage}
\usepackage{float}
\usepackage{bm}
\usepackage{algorithm}
\usepackage{amsmath}
\usepackage{amsthm}
\usepackage{amssymb}
\usepackage{amsfonts}
\usepackage[inline]{enumitem}
\usepackage{array}
\usepackage{multirow}
\usepackage{bbm}
\usepackage{enumitem}
\usepackage{dashrule}
\usepackage{tikz}
\usetikzlibrary{shapes.geometric, arrows}
\usepackage{caption}
\usepackage{changepage}
\usepackage{setspace}
\usepackage[en-US]{datetime2}
\usepackage[noadjust]{cite}
\usepackage{makecell}

\usepackage{threeparttable}

\usepackage{soul,xcolor}
\definecolor{DarkBlue}{rgb}{0.2,0.2,0.6}
\definecolor{DarkRed}{rgb}{0.368,0.097,0.078}
\usepackage{url}

\renewcommand{\paragraph}[1]{\medskip\noindent{\bf #1}}

\newtheorem{theorem}		{Theorem} 	[section]
\newtheorem{lemma}			[theorem]	{Lemma}
\newtheorem{definition}		[theorem]	{Definition}

\newtheorem{observation}		[theorem]	{Observation}
\newtheorem{corollary}		[theorem]	{Corollary}
\newtheorem{claim}			[theorem]	{Claim}
\newtheorem{question}		[theorem]	{Question}

\newtheorem{remark}		    [theorem]	{Remark}

\newcommand{\remove}[1]{}

\DeclareMathSymbol{\N}{\mathbin}{AMSb}{"4E}
\DeclareMathSymbol{\Z}{\mathbin}{AMSb}{"5A}
\DeclareMathSymbol{\R}{\mathbin}{AMSb}{"52}
\DeclareMathSymbol{\Q}{\mathbin}{AMSb}{"51}
\DeclareMathSymbol{\erert}{\mathbin}{AMSb}{"50}
\DeclareMathSymbol{\I}{\mathbin}{AMSb}{"49}
\DeclareMathSymbol{\C}{\mathbin}{AMSb}{"43}
\DeclareMathOperator{\E}{E}
\DeclareMathOperator{\Var}{Var}

\newcommand{\Geom}{{\rm Geom}}
\newcommand{\Bin}{{\rm Bin}}
\newcommand{\Binomial}{{\rm Bin}}
\newcommand{\Bernoulli}{{\rm Bernoulli}}

\newcommand{\AAA}{\mathcal A}

\newcommand{\PPP}{\mathcal P}

\newcommand{\eps}{\varepsilon}

\newcommand{\polylog}{\mathop{\rm polylog}}

\title{A Simple and Robust Protocol for Distributed Counting}

\author{
  Edith Cohen\thanks{Google Research and Tel Aviv University. Partially supported by Israel Science Foundation (grant 1156/23). \texttt{edith@cohenwang.com}}
  \and
    Moshe Shechner\thanks{Tel Aviv University. Partially supported by the Israel Science Foundation (grant 1419/24). \texttt{moshe.shechner@gmail.com}}
  \and
  Uri Stemmer\thanks{Tel Aviv University and Google Research. Partially supported by the Israel Science Foundation (grant 1419/24). \texttt{u@uri.co.il}}
}
\date{September 6, 2025}

\begin{document}
\maketitle

\begin{abstract}
We revisit the {\em distributed counting} problem, where a server must continuously approximate the total number of events occurring across $k$ sites while minimizing communication. The communication complexity of this problem is known to be $\Theta(\frac{k}{\eps}\log N)$ for deterministic protocols. Huang, Yi, and Zhang (2012) showed that randomization can reduce this to $\Theta(\frac{\sqrt{k}}{\eps}\log N)$, but their analysis is restricted to the {\em oblivious setting}, where the stream of events is independent of the protocol's outputs.

Xiong, Zhu, and Huang (2023) presented a {\em robust} protocol for distributed counting that removes the oblivious assumption. However, their communication complexity is suboptimal by a $\polylog(k)$ factor and their protocol is substantially more complex than the oblivious protocol of Huang et al.\ (2012). This left open a natural question: could it be that the simple protocol of Huang et al.\ (2012) is already robust?

We resolve this question with two main contributions. First, we show that the protocol of Huang et al.\ (2012) is itself {\em not} robust by constructing an explicit adaptive attack that forces it to lose its accuracy. Second, we present a new, surprisingly simple, robust protocol for distributed counting that achieves the optimal communication complexity of $O(\frac{\sqrt{k}}{\eps} \log N)$. Our protocol is simpler than that of Xiong et al.\ (2023), perhaps even simpler than that of Huang et al.\ (2012), and is the first to match the optimal oblivious complexity in the adaptive setting.
\end{abstract}

\section{Introduction}\label{sec:intro}

In the {\em distributed counting problem}, there is a server (a.k.a.\ the {\em coordinator}) and $k$ sites. Throughout the execution, each site observes {\em events} occurring at different times. The goal of the server is to maintain an ongoing $\eps$-approximation of the number of events observed by all sites together, while minimizing the communication cost.\footnote{For consistency with prior work, throughout the introduction we measure communication complexity by the total number of messages sent during the execution. See Section~\ref{sec:bits} for communication in bits.}

Keralapura, Cormode, and Ramamirtham \cite{keralapura2006communication} presented a simple {\em deterministic} protocol for this problem with a communication cost of $O\left(\frac{k}{\eps}\log N\right)$, where $N$ is the total number of events. In their protocol, each site tracks its local event count and notifies the server whenever this count increases by a factor of $(1+\eps)$. The number of notifications per site is bounded by $O(\frac{1}{\eps}\log N)$, and the total communication is bounded by $O(\frac{k}{\eps}\log N)$. Notably, this protocol requires only one-way communication from the sites to the server. 
Yi and Zhang \cite{yi2009optimal} later established a lower bound of $\Omega\left(\frac{k}{\eps}\log N\right)$ on the communication complexity of every  deterministic protocol, even when two-way communication is allowed. This shows that the protocol of \cite{keralapura2006communication} is  optimal among deterministic protocols. 

For {\em randomized} protocols, the picture is more nuanced. Huang, Yi, and Zhang \cite{huang2012randomized} showed that with only one-way communication, the complexity remains $\Theta(\frac{k}{\eps}\log N)$, offering no asymptotic improvement over deterministic protocols. However, for the two-way setting, they presented a protocol achieving a lower cost of $O(\frac{\sqrt{k}}{\eps}\log N)$, provided that $\eps \leq 1/\sqrt{k}$.\footnote{Throughout the introduction, when stating the communication complexity of randomized protocols we assume that the answers are pointwise accurate with constant probability. For example, for every time step $t$, the released answer is $\eps$-correct with probability at least $0.99$. The dependency on the failure probability will be made precise in the technical sections that follow.}
Furthermore, they proved a matching lower bound, showing that their two-way randomized protocol is asymptotically optimal among all randomized protocols for this problem.

\paragraph{Oblivious vs.\ robust protocols.} 
Huang et al.\ \cite{huang2012randomized} analyzed their randomized protocol under the assumption that the entire input sequence is fixed in advance. That is, for any fixed input sequence, their protocol must succeed with high probability over its internal randomness. This model is known as the {\em oblivious setting} because the entity generating the inputs is ``oblivious'' to the protocol's outputs. Protocols designed to succeed in this setting are called {\em oblivious} protocols.
In contrast, protocols that provably maintain utility even when the input sequence is chosen adaptively, as a function of the protocol's previous outputs, potentially in an adversarial manner, are called {\em robust} protocols (a.k.a.\  {\em adaptive} protocols). 
It can be easily seen that any deterministic protocol that guarantees correctness in the oblivious setting is automatically robust.\footnote{Let $\PPP$ be a deterministic oblivious protocol and suppose towards contradiction that there is an adaptive adversary $\AAA$ that causes $\PPP$ to fail on inputs generated (adaptively) by $\AAA$. Let $\vec{x}$ denote the input sequence generated by $\AAA$ when interacting with $\PPP$. Now, as $\PPP$ is deterministic, it must also fail on this input sequence $\vec{x}$ even when it is fixed ahead of time, which contradicts the utility guarantees of $\PPP$.} However, this is not the case for randomized protocols.
Intuitively, the difficulty is that when inputs are chosen adaptively, these inputs can become dependent on the internal randomness of the protocol, thereby breaking the analysis (and correctness) of many oblivious protocols.  
Designing (randomized) robust algorithms and data structures that outperform deterministic ones is a very active research area in several sub-fields of theoretical computer science.\footnote{See, e.g., \cite{holm2001poly,NanongkaiS17,Wulff-Nilsen17,NanongkaiSW17,chuzhoy2020deterministic,BernsteinC16,Bernstein17,BernsteinChechikSparse,ChuzhoyK19,ChuzhoyS20,gutenberg2020decremental,gutenberg2020deterministic,GutenbergWW20,Chuzhoy21,BhattacharyaHI15,BhattacharyaHN16,BhattacharyaHN17,BhattacharyaK19,Wajc19,BhattacharyaK21deterministic,ben2020framework,HassidimKMMS20,woodruff2020tight,alon2021adversarial,KaplanMNS21,Braverman2021adversarial,ACSS21,BeEO21,chakrabarti2022adversarially,HardtU14,dwork2015preserving,bassily2021algorithmic,steinke2017tight,gupta2021adaptive}.}

The work of Huang et al.\ \cite{huang2012randomized} gave rise to the question of understanding the communication complexity of {\em robust} protocols for the distributed counting problem. 
This question was addressed by Xiong, Zhu, and Huang \cite{xiong2024adversarially} who presented a robust variant of the protocol of Huang et al.\ \cite{huang2012randomized}. 
Their protocol is based on a generic ``robustification technique'' from Hassidim et al.\ \cite{HassidimKMMS20}. This technique uses differential privacy (DP) to protect not the input data, but rather the algorithm's internal randomness. As \cite{HassidimKMMS20} showed,
this can be used to limit (in a precise way) the dependencies between the internal randomness of
the algorithm and its inputs, thereby making it easier to argue about the adaptive setting.\footnote{Following \cite{HassidimKMMS20}, differential privacy was used as a tool to ``robustify'' many oblivious algorithms in several settings. See, e.g.,\  \cite{ACSS21,BeEO21,gupta2021adaptive,BeimelKMNSS22,Cohen0NSSS22,BrandGJLLPS22,SadigurschiSS23,CherapanamjeriS23}.} While this ``robustification technique'' is quite generic, it often comes at the cost of increased algorithmic complexity and reduced performance compared to the ``base'' oblivious algorithm. 

At a high level, the robust protocol of Xiong et al.\ \cite{xiong2024adversarially} has two components: (1) An oblivious protocol, similar to that of Huang et al.\ \cite{huang2012randomized} but more suited to serve as the ``base'' protocol in the robustification-via-DP framework; and (2) a DP layer that adds appropriately calibrated noise in order to ``hide'' the internal randomness used by the base protocol from anyone who observes the released outputs. This approach has three shortcomings:
\begin{enumerate}
    \item The resulting protocol is significantly more complex compared to the oblivious protocol of Huang et al.\ \cite{huang2012randomized}, mainly due to the DP layer.
    \item The use of differential privacy as a robustification method inflates the communication complexity by a factor of $\log^{0.75}(k)$, from $O\left(\frac{\sqrt{k}}{\eps}\log N\right)$ to $O\left(\frac{\sqrt{k}\cdot \log^{0.75}(k)}{\eps}\log N\right)$.
    \item The resulting protocol is robust only in the {\em black-box setting}, where the inputs might depend on the outputs of the server, but not directly on the internal state of the protocol or the messages transmitted between the server and the sites.  
\end{enumerate}

These shortcomings raise the following question:

\begin{question}\label{q:intro}
Could it be that the randomized protocol of Huang et al.\ \cite{huang2012randomized} is itself robust, thereby avoiding these shortcomings?
\end{question}

This question was left open by Xiong et al.\ \cite{xiong2024adversarially}. That is, even though they presented a robust variant of the protocol of \cite{huang2012randomized}, they did not present an attack or argue about the robustness of the original (randomized) protocol.

\subsection{Our results}
We give two answers to Question~\ref{q:intro}.
\begin{enumerate}
    \item First, we show that the original (randomized) protocol of \cite{huang2012randomized} is itself {\em not} robust. Specifically, we design an adaptive attack that generates the inputs (sequences of events) in a way that forces the protocol of \cite{huang2012randomized} to fail. 
    \item Nevertheless, we show that small modifications to the protocol of \cite{huang2012randomized} do make it robust (without using any external ``robustification techniques''). This allows us to avoid all three shortcomings listed above. Specifically, our protocol is as simple as that of \cite{huang2012randomized} (perhaps even simpler), our communication complexity is optimal, and our protocol is robust even in the face of a {\em white-box} attacker where the inputs might be generated both as a function of the released outputs and as a function of the internal state and messages transmitted between the server and the sites.
\end{enumerate}

These results are specified in the following two theorems.
\begin{theorem}[Adaptive attack, informal version of Theorem~\ref{thm:attack}]\label{thm:introAttack}
There exists an attack that generates adaptive inputs to the protocol of \cite{huang2012randomized} that forces it to fail (have error greater than $(1+\eps)$) with high probability.
\end{theorem}
\begin{theorem}[Simple robust protocol, informal version of Theorem~\ref{thm:RobustProtocol}]\label{thm:introProtocol}
There is a robust protocol in the white-box setting for the distributed counting problem with $k$ sites that guarantees $\eps$-accuracy over $N$ events while achieving optimal communication complexity of  $O\left(\frac{\sqrt{k}}{\eps}\log N\right)$.
\end{theorem}

Before giving a technical overview of our results, we first describe the oblivious protocol of \cite{huang2012randomized}, so that we could illustrate our adaptive attack on it. We then describe the robust protocol of \cite{xiong2024adversarially}, so that we could convey the simplicity of our new protocol.

\subsection{Informal overview of the oblivious protocol of \cite{huang2012randomized}}

At a high level, the protocol of \cite{huang2012randomized} operates in \emph{rounds}, where each round processes roughly twice the number of events as the previous one. Hence, the total number of rounds is \( \Theta(\log N) \).  The breakdown into rounds is determined by a deterministic background protocol, similar to that of Keralapura et al.\ \cite{keralapura2006communication}: Each site $i$ transmits a message to the coordinator each time its local event count $n_i$ doubles. This allows the coordinator to maintain a factor $2$ approximation $n'\in [n/2,n]$ of the total event count $n=\sum_{i\in[k]} n_i$. When $n'$ roughly doubles, the server ends the current round and broadcasts to all sites an updated \emph{transmission probability} 
$p \propto \frac{\sqrt{k}}{\varepsilon n'} = \Theta\left(\frac{\sqrt{k}}{\varepsilon n}\right)$. 
This transmission probability remains fixed during the round and is roughly halved from round to round.

The finer approximation is achieved via the following randomized protocol. 
For each event arrival at site $i$, the site samples a Bernoulli random variable with parameter $p$. If the sample is 1, the site transmits its exact local counter $\bar{n}_i =  n_i$ to the server. 
Throughout the execution, the server estimates each local count $n_i$ as $\hat{n}_i=0$ if no transmission was received from site $i$ and as $\hat{n}_i=\bar{n}_i - 1 + 1/p$, where $\bar{n}_i$ is the \emph{last} value received from site $i$.
The correction term $(-1+1/p)$ accounts for the expected value of $n_i-\bar{n}_i \sim \mathrm{Geom}[p]-1$, which is the number of events that occurred since the last report by the site. The server then estimates the total event count as $\hat{n} = \sum_{i \in [k]} \hat{n}_i$.
  
Huang et al.\ \cite{huang2012randomized} showed that for any fixed input sequence, the estimator for each local count, \( \hat{n}_i \), is unbiased, i.e., $\E[\hat{n}_i] = n_i$, 
and that it has a bounded variance of 
$\text{Var}(\hat{n}_i) \leq \frac{1}{p^2}$.
Thus, the variance of the final estimator satisfies
$\text{Var}(\hat{n}) = \text{Var}\left( \sum_{i \in [k]} \hat{n}_i \right) \leq \frac{k}{p^2}$, which is at most $(\eps n)^2$ as $p = \Theta\left(\frac{\sqrt{k}}{\varepsilon n}\right)$. So the standard deviation is of the order of $\eps n$, as required. 

When a round ends and $p$ is updated to $p'$, the server performs a random correction on its stored count $\bar{n}_i$ for each site $i$. This ensures that the corrected value has the exact same distribution as if the protocol had been run all along with transmission probability $p'$. Consequently, the variance bound ($\text{Var}(\hat{n})\leq\eps n$) remains valid throughout the execution.\footnote{ Specifically, when $p$ is updated to $p'$, the server updates every $\bar{n}_i$ to $\bar{n}_i' = \bar{n}_i -Z$ for an appropriate random variable $Z$ whose distribution depends on $p,p'$.}

\medskip

The expected communication cost per round includes \( O(k) \) messages for broadcasting the transmission probability \( p \), and \( p \cdot O(n) \) messages from the sites to the server, since each of the \( n \) events is transmitted with probability \( p \). Overall, the communication cost is
$$
O\left( \left(  k+pn  \right)\cdot\log N\right)=O\left( \frac{\sqrt{k}}{\eps}\log N\right),
$$
where the last equality follows by plugging in $p$ and by the assumption that \( k \leq 1/\varepsilon^{2} \).

\subsection{Our adaptive attack on \cite{huang2012randomized}}

We construct an adaptive adversary that observes the server's output and selectively generates events that cause the protocol of \cite{huang2012randomized}, which we denote as \( \AAA \), to fail.
The attack operates iteratively in a round-robin fashion over the sites as follows: it continues sending events to a given site until it observes a change in the output of \( \AAA \). Once a change is detected, the attack proceeds to the next site.
The intuition is that when the output of $\AAA$ changes, it means that the current site $i$ has just transmitted a message to the server, with $\bar{n}_i=n_i$. At that point the estimator \( \hat{n}_i \) is larger than the actual local count \( n_i \) by $(-1+1/p)$. Hence, by stopping sending events to that site, the attacker introduces a positive estimation bias. 
Furthermore, this bias persists (in expectation) even when the server changes to the next round, as the random corrections to $\bar{n}_i$ do not change the expectation of the estimate $\hat{n}_i$.
By repeating this process across multiple sites, the adversary accumulates a global bias in the total estimate \( \hat{n} \), thereby causing the server to lose accuracy.  After sufficiently many events, the expected bias at any given time is 
$\approx (k-1)/p \approx \sqrt{k} \eps n$. 
This can be formalized to obtain Theorem~\ref{thm:introAttack}. 
See Section~\ref{sec:resultAttackOGS} for more details.

\subsection{Informal overview of the robust protocol of \cite{xiong2024adversarially}}

We now elaborate on the robust protocol of Xiong et al.\ \cite{xiong2024adversarially}. As we mentioned, this protocol has two components: (1) An oblivious ``base protocol'', similar to that of \cite{huang2012randomized}; and (2) a DP-layer for robustifying this oblivious protocol.

We first describe the oblivious ``base protocol'' used by \cite{xiong2024adversarially}.
At a high level, like  \cite{huang2012randomized}, this oblivious protocol operates in rounds. In the beginning of each round, the server notifies the sites of the new round, collects exact counts from all the sites, and broadcasts the exact total count till now, denoted as $n_0$. Given $n_0$, every site $i$ divides its stream of events into blocks of size $\Delta = O(\varepsilon\cdot n_0 /    \sqrt{k})$. For each block $j$, site $i$ samples a random threshold $r_{i,j}\in[\Delta]$. Throughout the round, site $i$ sends a message to the server whenever its local count crosses any of its internal thresholds $r_{i,j}$. 
On the server side, we estimate the total event count from the beginning of the round as $B\cdot\Delta$, where $B$ denotes the number of messages received from all sites together. 
This protocol has similar performance to \cite{huang2012randomized}, and similarly, its analysis assumes an oblivious input.

We now elaborate on how Xiong et al.\ \cite{xiong2024adversarially} robustified this protocol using DP. To motivate this, observe that if the attacker knows all the internal thresholds $r_{i,j}$, then it can easily attack this protocol, similarly to the way we attacked the protocol of \cite{huang2012randomized}. Specifically, the attacker sends events to site $i$ till its local count crosses a ``large'' threshold, say at least $3\Delta/4$. This accumulates a bias of $\Omega(\eps n_0/\sqrt{k})$ in the estimation for the local count of site $i$. The attacker then continues to the next site. After $\Omega(\sqrt{k})$ sites, the attacker achieves a significant bias of more than $\eps n_0$, thereby causing the protocol to fail. So we want to hide the internal thresholds from the attacker.

However, note that even if the attacker does not know the thresholds then it can still attack the protocol in a similar manner. Specifically, the attacker sends events to site $i$ till it sees that the output of the protocol was modified. At that moment the attacker knows that site $i$ has just transmitted a message, and therefore it learns its current threshold exactly (as the attacker knows how many events site $i$ has received), and it can hence conduct the same attack as before. So we also want to hide the times at which the output of the protocol changes.

Xiong et al.\ \cite{xiong2024adversarially} used two generic techniques from differential privacy in order to hide these two aspects of their oblivious protocol:  the \texttt{AboveThreshold} algorithm of \cite{dwork2009complexity} and the binary tree mechanism of \cite{DworkNPR10}. They then proved, building on the work of Hassidim et al.\ \cite{HassidimKMMS20}, that the resulting protocol is robust in the {\em black-box} model. On the other hand, a {\em white-box} attacker, who knows the internal thresholds or sees the communication between the server and the sites, can still conduct the attack described above. Thus, the protocol of \cite{xiong2024adversarially} is not robust in the white-box model.

\subsection{Informal overview of our simple robust protocol}

Our protocol also operates in rounds. Similarly to the oblivious ``base protocol'' of \cite{xiong2024adversarially}, in the beginning of each round, the server notifies the sites of the new round,  collects exact counts from all the sites, and broadcasts the exact total count so far, denoted as $n_0$. Given $n_0$, every site determines this round's {\em transmission probability} $p\approx \frac{\sqrt{k}}{\eps n_0}$. Then, throughout the round, whenever a site receives an event, it samples a bit $b$ from $\Bernoulli(p)$ and if $b=1$ then the site sends the message ``1'' to the server. At any moment throughout the round, the server estimates the current total number of events as $n_0+B/p$, where $B$ denotes the number of ``1'' messages the server has received from the beginning of the round. The round ends when $B=k$, which happens after $\frac{k}{p}\approx\eps\sqrt{k}n_0$ events in expectation.

\medskip

We show that this absurdly simple protocol is both optimal and adversarially robust (in the white-box model). We sketch the arguments here;  see Sections~\ref{sec:ourRobustProtocol},~\ref{sec:analysisRobust},~\ref{sec:perevent} for the formal details.

\paragraph{Communication.} 
Note that $n_0$ increases (in expectation) by a factor of $(1+\eps \sqrt{k})$ from round to round. Therefore, the total number of rounds is $\approx \log_{1+\eps\sqrt{k}} N$, which for $\eps< 1/\sqrt{k}$ is $\approx \frac{1}{\sqrt{k}\eps}\log N$. The communication cost of each round is $O(k)$: In the beginning of the round the server broadcasts the start of the new round and obtains an exact local count from each of the sites, and throughout the round the server receives exactly $k$ ``1'' messages from the sites. Therefore, the expected overall communication cost is
$O\left( \frac{\sqrt{k}}{\eps}\log N\right)$, matching the communication cost of the oblivious protocol of \cite{huang2012randomized}.

\paragraph{Robustness.} The robustness of our protocol follows from the fact that it is symmetric w.r.t.\ which site receives an event. Specifically, at any given moment throughout the execution, and for every conditioning on the transcript of the interaction so far, the distribution on the outcome of the server is {\em identical} whether site $i$ receives an event or site $j$ receives an event. Thus, the attacker might as well send all events to site $i=1$, as this has no effect on the outcome distribution of the server. As sending all events to site $i=1$ is an {\em oblivious} stream of inputs, it suffices to analyze the error of our protocol in the oblivious setting.

\paragraph{Accuracy.}
So let us consider the oblivious stream where all events arrive at site $i=1$. For consistency with prior works, we analyze our protocol in terms of {\em ``per-time''} accuracy (a.k.a. {\em pointwise} accuracy): For every time step $n$, with probability 0.99 the returned estimate at this time, $\hat{n}$, satisfies $|\hat{n}-n|\leq\eps n$. Here we only sketch the argument; see Sections~\ref{sec:analysisRobust} and~\ref{sec:perevent} for the formal details, along with a {\em ``for-all''} accuracy guarantee (uniform across all time steps). At a high level, our utility analysis has two parts:
\begin{enumerate}
    \item[{\bf (1)}] {\bf {\em ``Per-round''} accuracy.} For every $r$ we show that with probability at least 0.99, the maximum relative error of the count at all events that fall in the $r$th round is at most $\eps$. To see why this holds, recall that the number of ``1'' messages we receive during the round is exactly $k$, and that the number of events witnessed during this round is distributed as the sum of $k$ geometric random variables $X_1,\dots,X_k\sim\Geom(p)$. 
    That is, after the $\ell$th ``1'' message, our estimate is $\left(n_0+\frac{\ell}{p}\right)$, where $n_0$ is the exact count at the beginning of the round, while the actual event count is $\left(n_0+\sum_{j=1}^{\ell} X_{j}\right)$.
    We can thus bound our estimation error using standard (partial-sum) tail bounds on the geometric distribution. Specifically, with probability at least 0.99 it holds that $\max\limits_{1\le \ell\le k}\left|\sum_{j=1}^{\ell} X_{j}-\frac{\ell}{p}\right|\lesssim\frac{\sqrt{k}}{p}\approx \eps n_0\leq\eps n$, yielding our desired relative error throughout the round.

    \item[{\bf (2)}] {\bf From per-round to per-time accuracy.} 
    Converting a per-round guarantee into a per-time guarantee is subtle, because for a fixed time step $n$ 
    it is not clear a priori which round $r$ the time step $n$ will belong to. Furthermore, conditioning on $n$ belonging to a specific round $r$ changes the distribution of our geometric RVs and breaks our per-round analysis. A naive solution here would be to union bound over all rounds, but this would yield a ``for-all'' guarantee that would not match the optimal per-time complexity. Informally, we overcome this by showing that it suffices to union bound over {\em constantly} many rounds rather than all rounds. More specifically, for a fixed time step $t$ we let $t_0<t$ denote the largest time step such that with high probability there is at least one round that starts between time $t_0$ and $t$. This allows us to ignore the execution before time $t_0$, as the error in step $t$ is independent of past rounds (due to the sync in the beginning of the round starting between $t_0$ and $t$). Furthermore, we show that in expectation there is at most a constant number of rounds that begin between $t_0$ and $t$, and that it suffices to argue only about these $O(1)$ rounds.
\end{enumerate}

These arguments can be formalized to obtain Theorem~\ref{thm:introProtocol}, showing that our simple protocol is both optimal and adversarially robust in the white-box model.

\subsection{Empirical demonstration}
In \cref{sec:experiments}, we present simulation results comparing our robust protocol with the oblivious protocol of \cite{huang2012randomized}. We test both protocols on two types of streams: our adaptive attack stream and a non-adaptive input stream. We did not include an evaluation of the robust protocol of \cite{xiong2024adversarially}, as its source code was not available, its implementation relies on external DP libraries, and its associated constants appear to be large. We observe that:
\begin{enumerate}
\item[(1)] As predicted by our analysis, the protocol of \cite{huang2012randomized} is vulnerable to the adaptive attack, whereas our protocol is robust. The difference is quite striking, even more so than guaranteed by our theoretical analysis of the attack, which suggests the constants in our analysis were not tight.

\item[(2)] On the non-adaptive stream, both protocols exhibit similar performance. This shows that robustness does not come at the expense of practicality. (Recall that both protocols have the same asymptotic guarantees on oblivious streams.)

\end{enumerate}

\subsection{Related works}

The distributed counting problem is a specific instance of the more general framework of {\em distributed functional monitoring}, where the goal is to continuously track a function over data streams distributed across multiple sites. This framework has been extensively studied for various functions beyond simple sums. Significant research has focused on estimating frequency moments, identifying heavy hitters, approximating quantiles, and more \cite{dilman2002efficient, cormode2005holistic, 
keralapura2006communication,
CormodeMY08,
arackaparambil2009functional, tirthapura2011optimal, cormode2012continuous, woodruff2012tight, Cormode13,chen2017improved, wu2020learning,HuangXZW25}.

The distributed functional monitoring model is related to other multi-player communication models. A notable example is the classic {\em coordinator model}, introduced by \cite{dolev1989multiparty}, where $k$ players, each holding a static input, communicate with a central referee to compute a function of their joint inputs in a single shot. A significant body of work has analyzed the communication and round complexity of core problems within this model, including \cite{phillips2012lower,
woodruff2013distributed, 
woodruff2014optimal,
viola2015communication,
braverman2017rounds, 
huang2017communication,
huang2021communication,
assadi2024rounds,
esfandiari2024optimal}.

%%%%%%%%%%%%%%%%%%%%%%%%%%%%%%%%
\section{An Adaptive Attack on the HYZ12 Protocol} \label{sec:resultAttackOGS}

In this section we show that the \textsf{HYZ12}, the protocol proposed in \cite{huang2012randomized}, is not robust by presenting an adaptive attack.  
The attack consistently induces a relative estimation error larger by a factor of $\Omega(\sqrt{k})$ compared to the target accuracy~$\varepsilon$.

\begin{theorem}[Attack on \textsf{HYZ12}] \label{thm:attack}
There exist universal constants $c_1,c_2,c_3>0$ and an adaptive attack on the protocol of \cite{huang2012randomized} such that for every event count $n > c_3 \sqrt{k}/\varepsilon$,
\[
\Pr\!\left[\hat{n}-n > c_1 \sqrt{k}\,\varepsilon n\right] \;\ge\; 1 - e^{-c_2 k}.
\]
\end{theorem}

\subsection{Description of \textsf{HYZ12}}\label{sec:descriptionHYZ12}

The \textsf{HYZ12} protocol \cite{huang2012randomized} is described in 
\cref{alg:HYZ12,alg:doublingdet}. We briefly review the protocol and its analysis, which assumes that the \emph{event stream is fixed in advance}.

The \textsf{Doubling protocol} (\cref{alg:doublingdet} \cite{keralapura2006communication}) is
deterministic and runs concurrently on the same input.  
It maintains a constant-factor approximation of the total event count: each site notifies the server whenever its local count doubles, and the server triggers a $\textsf{BoundaryReached}$ alert when its global estimate doubles. On $N$ events, this protocol incurs a total communication cost of $k\log N$.

The site-side \textsf{HYZ12} protocol is very simple: Each site $i$ keeps a local copy of a transmission probability $p$ that is updated via server broadcasts and its local event count $n_i$. On each event, it reports its local count to the server with probability $p$.
On the server side, the protocol progresses in rounds. The server stores for each site a value $\bar n_i$ (the last reported count, adjusted as described below) and computes the per-site estimate $\hat n_i = \bar n_i + 1/p - 1$ (when $\bar n_i > 0$). The global estimate is $\hat n = \sum_i \hat n_i$. A new round starts upon a $\textsf{BoundaryReached}$ alert, at which point the server broadcasts an updated probability $p$ and adjusts each $\bar n_i$ so that the resulting state matches the distribution that would have arisen had the new $p$ been in effect from the beginning.

The analysis of \cite{huang2012randomized} holds under the assumption that
the event stream is fixed in advance. 
First, note that with a fixed transmission probability $p$, the most recent report from a site with current local count $n_i$ satisfies
\[
\bar n_i \;\sim\; n_i + 1 - \Geom[p].
\]
Hence the per-site estimate $\hat n_i := \bar n_i + 1/p - 1$ is an unbiased estimator of $n_i$.  

When the protocol decreases $p$ at the start of a new round, the server updates each $\bar n_i$ by subtracting a random variable $Z_i$ with the property\footnote{These are \emph{Zero-inflated geometrics}, see \cref{def:zeroinflated}.} that
\[
\Geom[p_{\text{old}}] + Z_i \;\stackrel{d}{=}\; \Geom[p].
\]
This ensures that the distribution of the server state after the update is identical to the distribution that would have been obtained had the protocol run with the new probability $p$ from the start. Moreover, the expected value of $\hat n_i$ is unchanged, since $\E[Z_i] = 1/p - 1/p_{\text{old}}$.

Therefore, under a fixed input stream, the global estimate $\hat n=\sum_i \hat n_i$ remains unbiased for all time.  
The error distribution at any event count $n$ is stochastically dominated by the deviation of a sum of $k$ independent $\Geom[p]$ random variables from its expectation $k/p$.  
By standard sub-exponential tail bounds (see \cref{eq:bernstein}), for any fixed $\delta>0$, the deviation is $O(\varepsilon n)$ with probability at least $1-\delta$.

\medskip
This completes our overview of the analysis of \cite{huang2012randomized} in the non-adaptive setting.  
We demonstrate that these guarantees fail under adaptively chosen event streams.

\begin{algorithm2e}[H]
\DontPrintSemicolon
\caption{\textsf{Doubling Protocol}  \cite{keralapura2006communication}\label{alg:doublingdet}}
\KwOut{A running estimate $n' \in [n/2, n]$ of $n=\sum_i n_i$ and \textsf{BoundaryReached} when $n'$ doubles.}
\BlankLine
\textbf{Server state:} for each site $i$, $n'_i \gets 0$; total $n' \gets 0$; doubling threshold $\tau \gets 1$.\;
\textbf{Site $i$ state:} local count $n_i \gets 0$; \;
\BlankLine
\textbf{Event arrival at site $i$:}\tcp*{runs independently at each site}
\Indp
$n_i \gets n_i + 1$\;
\lIf{$n_i \textrm{ is a power of } 2$}{send \textsf{DoublingNotify}$(i,n_i)$ to server}
\Indm
\BlankLine

\textbf{Upon server receiving \textsf{DoublingNotify}$(i,n_i)$:}\;
\Indp
$n'\gets n' + n_i-n'_i$\, ; $n'_i \gets n_i$\;
\If{$n' \ge 2\tau$}{\textsf{BoundaryReached}$(n')$; $\tau \gets n'$}\tcp*{$n'$ doubled}
\Indm
\end{algorithm2e}

\begin{algorithm2e}[H]
\DontPrintSemicolon
\caption{\textsf{HYZ12 Protocol} \cite{huang2012randomized}\label{alg:HYZ12}}
\KwIn{parameters $k,\varepsilon$; \textsf{BoundaryReached} from protocol \cref{alg:doublingdet}.}
\KwOut{Running estimate $\hat n = \sum_i \hat n_i$.}
\BlankLine
\tcp{\textbf{Site protocol for sites $i\in[k]$:}}
\textbf{Site $i$ state:} $p\gets 1$ (updated via broadcasts by server), $n_i\gets 0$ event counter.\;
\BlankLine
\textbf{Event arrival at site $i$:}\tcp*{runs independently at each site}
\Indp
$n_i\gets n_i+1$\;
Sample $X\sim \Bernoulli(p)$.\;
\lIf{$X=1$}{send \textsf{Report}$(i, n_i)$ to server\tcp*{report exact local $n_i$}}
\Indm
\textbf{Upon site $i$ receiving \textsf{Broadcast}$(p)$:}\quad update local $p$.\;
\BlankLine
\tcp{\textbf{Server protocol:}}
\textbf{Server state:} for each site $i$: last report $\bar n_i \gets 0$, per-site estimate $\hat n_i \gets 0$; current transmit prob $p \gets 1$.\;
\BlankLine
\textbf{Upon \textsf{BoundaryReached}$(n')$ OR server receiving \textsf{Report}$(i,\bar n)$:} \tcp*{If an event triggers both, server processes the report first.}
\Indp
\If{\textsf{Report}$(i,\bar n)$ is received}{
$\bar n_i \gets \bar n$;\quad $\hat n_i \gets \bar n_i - 1 + \frac{1}{p}$.\; Publish $\hat n \gets \sum_j \hat n_j$.\;}
\If(\tcp*[f]{close old round, open new round}){\textsf{BoundaryReached}$(n')$ }{
$p_{\text{old}}\gets p$;\quad $p \gets 2^{\min\{0, \lfloor \log_2 \frac{\sqrt{k}}{\varepsilon n'} \rfloor \}}$ \tcp*{power of $2$ that is $\Theta\big(\frac{\sqrt{k}}{\varepsilon n'}\bigg)$.}
\If(\tcp*[f]{When $p=1$ server maintains exact counts}){$p<1$}{
\For(\tcp*[f]{Adjust $\bar n_i$ to updated $p$}){$i\in[k]$}{
    Sample $Z_i \sim Z_{p,p_{\mathrm{old}}}$ \tcp*{as in \cref{def:zeroinflated}; 
    $B\sim \Bernoulli[1- p/p_{\text{old}}]$; $G\sim \Geom\!\left[p\right]$; $Z_i\gets B\cdot G$}
    $\bar n_i \gets \max\{0,\ \bar n_i - Z_i\}$.\;
    \leIf{$\bar n_i = 0$}{$\hat n_i \gets 0$}{$\hat n_i \gets \bar n_i - 1 + \frac{1}{p}$}
}
\textsf{Broadcast}$(p)$ to all sites.\;
}}
\Indm
\end{algorithm2e}

\begin{algorithm2e}[H] 
\DontPrintSemicolon
\caption{Attack on Protocol \textsf{HYZ12} \label{alg:attack}}
\KwIn{Access to protocol with $k$ sites  $i\in[k]$: \textsf{InjectEvent}$(i)$, \textsf{Observe}() returning current $\hat n$.}
\KwOut{A sequence of adversarial event injections.}
\BlankLine
Initialize $i\gets 1$; $\hat n_{\text{last}}\gets \textsf{Observe}()$.\;
\BlankLine
\While{horizon not reached}{
    \tcp{Focus on site $i$ only; all other sites are idle}
    \Repeat(\tcp*[f]{Wait for a change in running estimate}){%
        $\hat n \gets \textsf{Observe}() \neq \hat{n}_{\text{last}}$ \;
        }{\textsf{InjectEvent}$(i)$\tcp*{Generate event at site $i$}}
    $\hat{n}_{\text{last}}\gets \hat{n}$\tcp*{Record current estimate}
    $i\gets 1 + (i \bmod k)$\tcp*{Move to the next site in round robin}
}
\end{algorithm2e}

\subsection{Attack description}
Our attack on the HYZ12 protocol (\cref{alg:HYZ12}) is given in \cref{alg:attack}.  
The attack continuously monitors the estimates $\hat n$ published by the server.  
It generates the input stream by injecting events exclusively at the current site $i$ until the server updates its estimate.  
At that moment the attacker switches to the next site in a round-robin order.

\subsection{Preliminaries for the Analysis of the Attack} \label{sec:attacktools}

For brevity, we use the same symbol to denote a random variable and its distribution.  
We write $\Geom(q)$ for a geometric random variable with success probability $q$, supported on $\{1,2,\ldots\}$ with
$\Pr(G=k)=(1-q)^{k-1}q.$
We also use the standard notation $\Binomial(n,p)$ and $\Bernoulli(p)$ for the binomial and Bernoulli distributions, respectively.

\begin{lemma}[Chernoff bound, multiplicative lower tail {\cite{mitzenmacher2017probability}}]\label{thm:chernoffIneq}
Let $x_1,\dots,x_n$ be independent $\{0,1\}$ random variables, and set $X=\sum_{i=1}^n x_i$ with mean $\mu=\E[X]$. Then for all $0\le\alpha\le1$,
\[
\Pr[X \le (1-\alpha)\mu] \;\le\; \exp\!\left(-\tfrac{\mu \alpha^2}{2}\right).
\]
\end{lemma}

\begin{corollary}[Binomial lower tail]\label{col:chernoffBinTwoTails}
Let $B \sim \mathrm{Binomial}(n,p)$ with $p<\tfrac12$. Then for all $0\le \alpha \le 1$,
\[
\Pr[B \le (1-\alpha)pn] \;\le\; \exp\!\left(-\tfrac{pn \alpha^2}{2}\right).
\]
\end{corollary}

\begin{definition} [Zero-inflated geometric random variables]\label{def:zeroinflated}
    For $0\leq q\le p\leq 1$, define  
\[
Z_{q,p}\ :=\ B\cdot G,\qquad 
B\sim \Bernoulli\!\left(1-\frac{q}{p}\right),\ \ 
G\sim \mathrm{Geom}(q),
\]
with $B$ and $G$ independent.
\end{definition}

See \cref{sec:attacktoolsproofs} for proofs for the following two Lemmas.

\begin{restatable}{lemma}{GeometricTelescopes}[Zero-inflated geometric telescopes]\label{lemm:geometrictelescopes}
Fix parameters $1\ge p_1\ge p_2\ge \cdots \ge p_r>0$. 
If $Z_{p_{i+1},p_i}$, $i=1,\dots,r-1$, are independent, then
\[
\sum_{i=1}^{r-1} Z_{p_{i+1},p_i}\ \stackrel{d}{=}\ Z_{p_r,p_1}.
\]
\end{restatable}

\begin{restatable}{lemma}{zeroinflatedtail}[Bernstein tail for a sum of zero-inflated geometrics] \label{lemm:BersteinTailZeroInflated}
Let $p\in(0,1]$, $\alpha_1,\dots,\alpha_r\in[0,1]$, and for each $i$ let
$X_i=B_i\,G_i$ with $B_i\sim\Bernoulli(\alpha_i)$ and
$G_i\sim\mathrm{Geom}(p)$, supported on $\{1,2,\dots\}$ with
$\Pr(G_i=k)=(1-p)^{k-1}p$. Assume $\{(B_i,G_i)\}_{i=1}^r$ are independent.
Put $A:=\sum_{i=1}^r\alpha_i$ and $\mu:=\E\!\left[\sum_i X_i\right]=A/p$.
Then for all $t\ge 0$,
\[
\Pr\!\Bigl(\sum_{i=1}^r X_i-\mu\ge t\Bigr)
\;\le\;
\exp\!\left(-\,\frac{t^2}{2\left(\dfrac{A(1-p)}{p^2}+\dfrac{t}{p}\right)}\right)
\;\le\;
\exp\!\left(-\,\frac12\min\!\left\{\frac{p^2t^2}{A(1-p)},\,pt\right\}\right).
\]
\end{restatable}

We bound the configured transmission probability and the round length in terms of the event count $n_0$ at the start of a round. These bounds hold for every input stream, including adaptively chosen inputs.

\begin{observation}[transmission probability at round start]\label{obs:HYZ12_p_bound}
Let a new round be initiated when the total event count is $n_0$.  
Then the configured transmission probability $p$ satisfies
\[
\min\!\left\{1,\ \frac{\sqrt{k}}{2\varepsilon n_0}\right\}
\;\le\; p \;\le\;
\min\!\left\{1,\ \frac{2\sqrt{k}}{\varepsilon n_0}\right\}.
\]
\end{observation}

\begin{proof}
The server counter $n'$ in \cref{alg:doublingdet} satisfies $n'\in[n/2,n]$.
Write $x=\tfrac{\sqrt{k}}{\varepsilon n'}$.
If $x\ge1$ then $\lfloor\log_2 x\rfloor\ge0$ and $p=1$. %
If $x<1$ then $p=2^{\lfloor\log_2 x\rfloor}\in[x/2,x]$, so
$\tfrac{\sqrt{k}}{2\varepsilon n_0}\le p\le \tfrac{2\sqrt{k}}{\varepsilon n_0}$.
Both cases give the stated bounds.
\end{proof}

\begin{observation}[Relative round length]\label{obs:round-length}
Let $n_0 < n_1$ be the event count at the start of two consecutive rounds of \cref{alg:doublingdet}. Then
\[
\frac{n_1-n_0}{n_0} \;\le\; 7 .
\]
\end{observation}
\begin{proof}
Let $n'_i$ be the respective values of the server counter $n'$ at the start of the rounds. It holds that $n'_1 \in [2 n'_0,4n'_0]$.
Furthermore, at any given time, the server’s counter $n'$ satisfies $n'\in[n/2,n]$, where $n$ is the actual event count. Therefore $n_0 \geq n'_0$ and $n_1 \leq 2n'_1\leq 8n'_0\leq 8n_0$.
\end{proof}

\subsection{Analysis of the Attack}

Throughout this section we assume $\varepsilon \le 1/\sqrt{k}$.  
We analyze the distribution of the protocol execution under the attack.  
A \emph{transcript} denotes the full specification of the states of the attack, sites, and server up to some stopping point.

\begin{lemma}[Many reporting sites]\label{lem:attack-Sites-per-round}
Assume $\varepsilon\sqrt{k}<1$.  
Fix a transcript up to the initiation of a round at event count $n_0 \ge 4\sqrt{k}/\varepsilon$, and consider the execution from event $n_0+1$ onward.  
Let $R$ denote the number of distinct sites that send a \textsf{Report} message to the server during the next $n_0$ events. Then
\[
\Pr\!\left[R > \tfrac{k}{8}\right] \;\ge\; 1-\exp\!\left(-\tfrac{\sqrt{k}}{32\varepsilon}\right).
\]
\end{lemma}
\begin{proof}
Each update to the published estimate is caused by a \textsf{Report} message from the site currently injected by the attack. Whenever a site reports, the attack observes the change in $\hat{n}$ and moves to the next site.

The transmission probability values during the next  $n_0$ events are configured at a round that starts at some event count $n\in [n_0,2n_0]$ (that is, either at the round that starts at $n_0$ or at a later round).
 By \cref{obs:HYZ12_p_bound}, using our assumption that $n_0 \geq 4\sqrt{k}/\eps$, the transmission probability values are in the interval
\[
\left[\frac{\sqrt{k}}{4\varepsilon n_0}, \frac{2\sqrt{k}}{\varepsilon n_0}\right].
\] 

 Therefore
\[
\Pr\!\left[R<\tfrac{k}{8}\right] \;\le\;
\Pr\!\left[\Bin\left[n_0,\frac{\sqrt{k}}{4\varepsilon n_0} \right]<\tfrac{k}{8}\right].
\]

Applying the Chernoff bound \cref{col:chernoffBinTwoTails} with 
$\alpha = 1-\tfrac{\varepsilon\sqrt{k}}{2}$ gives
\[
\Pr\!\left[\Bin\!\left[n_0,\tfrac{\sqrt{k}}{4\varepsilon n_0}\right]<\tfrac{k}{8}\right]
\;\le\; \exp\!\left(-\tfrac{1}{2}\tfrac{\sqrt{k}}{4\varepsilon}\Bigl(1-\tfrac{\varepsilon\sqrt{k}}{2}\Bigr)^2\right)
\;\le\; \exp\!\left(-\tfrac{\sqrt{k}}{32\varepsilon}\right),
\]
where the last inequality uses $\varepsilon\sqrt{k}\le 1$.

Finally, by the round-robin structure of the attack, once $k/8$ sites have reported, they are necessarily distinct. Hence with probability at least $1-\exp(-\sqrt{k}/(32\varepsilon))$, the number of distinct reporting sites in the round exceeds $k/8$.
\end{proof}

\begin{corollary}[Many sites report once $n/14$ events have occurred]\label{cor:manyaftern-5}
Let $n \ge 56\sqrt{k}/\varepsilon$, and let $n_0 \ge n/14$ be the first event index that is greater or equal to $n/14$ at which a new round starts.  
Let $R^*_n$ be the number of distinct sites that report at least once during $[n_0,n]$.  
Then
\[
\Pr\!\left[R^*_n > \tfrac{k}{2}\right] \;\ge\; 1-\exp\!\left(-\tfrac{\sqrt{k}}{8\varepsilon}\right),
\]
where the probability is over the randomness of the protocol execution from $n_0+1$ onward, conditioned on the transcript up to $n_0$.
\end{corollary}

\begin{proof}
By \cref{obs:round-length}, some round must start between $n/14$ and $n/2$; let $n_0$ be the first such start.  
Applying \cref{lem:attack-Sites-per-round}, with probability at least $1-\exp(-\sqrt{k}/(8\varepsilon))$, more than $k/2$ sites report in the interval
$[n_0,2n_0] \subset [n_0,n]$.
\end{proof}

\begin{definition}[Reporting transcript]
Consider the randomness in the execution of the protocol under attack.  Observe that 
the attack algorithm is deterministic.  
The protocol, however, uses two types of randomness:
\begin{itemize}
    \item site-side randomness, determining when sites send reports, and
    \item server-side randomness, used to adjust the counters $\bar n_i$ at the beginning of each round.
\end{itemize}

We call a transcript that fixes all site-side randomness but leaves the server-side randomness unspecified a \emph{reporting transcript}.
\end{definition}

\begin{observation}
A reporting transcript uniquely determines the sequence of round start indices, the configured transmission probabilities, and the states of the attack algorithm and the sites.  
In particular, the value of $R^*_n$ is fixed by the reporting transcript.  
The only remaining randomness lies in the server-side adjustments of $\bar n_i$ and, consequently, in the estimates $\hat n$ published by the server.
\end{observation}

\begin{lemma}[Attack efficacy conditioned on reporting transcript]\label{lem:attack-efficacy}
Fix an event index $n \geq 500 \sqrt{k}/\eps$ and the reporting transcript up to that point.  
Let $m^*$ denote the number of events received by the active site at event $n$ since its most recent report.  
Assume $R^*_n \ge k/2$, where $R^*_n$ is as in \cref{cor:manyaftern-5}.  
Then there exist absolute constants $c_1,c_2>0$ such that
\begin{align}
\E\!\left[\hat n - n \right] 
  &= \E\!\left[\sum_{i=1}^k (\hat n_i - n_i)\right] 
  \;\;\ge\;\; c_1\sqrt{k}\varepsilon n - m^*, \label{eq:expectation}
\\[1ex]
\Pr\!\left[\hat n - n \;\ge\; \tfrac{c_1}{2} \sqrt{k}\,\varepsilon n -m^*\right] 
  &\;\;\ge\; 1 - e^{-c_2 k}. \label{eq:highprob}
\end{align}
\end{lemma}

\begin{proof}
Let $p$ be the transmission probability at event $n$.
For each site $i$, let $p'_i\!\ge p$ be the transmission probability in the round in which $i$ most recently reported.
By \cref{obs:round-length}, any event time $t\in[n/2,n]$ lies in a round that started at some $n_0(t)\ge t/7\ge n/14$.
Applying \cref{obs:HYZ12_p_bound} at that start time gives
\[
p'_i \;\le\; \frac{2\sqrt{k}}{\varepsilon\,n_0(t)}
\;\le\; \frac{28\sqrt{k}}{\varepsilon\,n}
,
\]
for every site that reports at some $t\in[n/2,n]$.
Since $R_n^*\ge k/8$, at least $k/8$ sites satisfy $p'_i \le \frac{28\sqrt{k}}{\varepsilon\,n}$.

Observe that $\bar n_i$ equals the last count reported by site $i$, minus the cumulative randomized adjustments applied so far.  
By \cref{lemm:geometrictelescopes}, the distribution of this cumulative adjustment is $Z_{p,p'_i}$.

For $i\neq i^*$, the last reported count is equal to the current count $n_i$ and therefore
\[
\bar n_i = n_i - Z_{p,p'_i},\qquad 
\E[\bar n_i]= n_i - \tfrac{1}{p} + \tfrac{1}{p'_i}.
\]
For the active site $i^*$, the last reported count is $n_i-m^*$ and therefore
\[
\bar n_{i^*} = n_{i^*} - m^* - Z_{p,p'_{i^*}},\qquad
\E[\bar n_{i^*}]= n_{i^*} - m^* - \tfrac{1}{p} + \tfrac{1}{p'_{i^*}}.
\]
With the definition $\hat n_i := \bar n_i + \tfrac{1}{p} - 1$, we have
\[
\E[\hat n_i] =
\begin{cases}
n_i + \tfrac{1}{p'_i} - 1, & i\ne i^*,\\[2pt]
n_{i^*} - m^* + \tfrac{1}{p'_{i^*}} - 1, & i=i^*.
\end{cases}
\]
Hence
\[
\E\!\left[\sum_{i=1}^k (\hat n_i - n_i)\right]
= \sum_{i=1}^k \frac{1}{p'_i} - k - m^*
\;\ge\; \frac{k}{8}\cdot \frac{\eps n}{28\sqrt{k}} - k - m^*
\;=\; \frac{\sqrt{k}\eps n}{224} - k - m^*.
\]

Since by our choice of large enough $n$, taking $c_1 = 1/500$ we obtain \cref{eq:expectation}.

For deviations, note that
\[
\sum_{i=1}^k \hat n_i
= \sum_{i=1}^k n_i - m^* + \sum_{i=1}^k \Bigl(\frac{1}{p'_i}-1\Bigr)
- \sum_{i=1}^k Z_{p,p'_i}.
\]
 Therefore, the distribution of the 
fluctuations of $\sum_i \hat n_i -n $ around its mean is exactly the distribution of the fluctuations of
$\sum_i Z_{p,p'_i}$ around (its mean) $\sum_i (1/p - 1/p'_i)$.
By \cref{obs:HYZ12_p_bound}, $p\geq \sqrt{k}/(2\varepsilon n)$.

Applying \cref{lemm:BersteinTailZeroInflated} with deviation $t=\tfrac{c_1}{2}\sqrt{k}\,\varepsilon n$, 
using $p\ge \sqrt{k}/(2\varepsilon n)$ and $A\le k$, yields
\[
\Pr\!\left(\left|\sum_{i=1}^k Z_{p,p'_i} - \sum_{i=1}^k \Bigl(\frac{1}{p}-\frac{1}{p'_i}\Bigr)\right|
> \tfrac{c_1}{2}\sqrt{k}\,\varepsilon n\right)
\;\le\; e^{-\frac{c_1^2}{8}k}.
\]
Rewriting in terms of $\hat n-n+m^*$ and choosing the constants appropriately gives the claimed high-probability bound \eqref{eq:highprob}.
\end{proof}

\begin{proof}[Proof of \cref{thm:attack}]
Fix an event index $n$ and consider a transcript of the execution under the attack up to that point.  
We call a transcript \emph{bad} if either $R^*_n < k/2$ or $m^* > \tfrac13 c_1 \sqrt{k}\varepsilon n$.  

From \cref{cor:manyaftern-5}, the first bad event has probability at most $\exp(-\Omega(\sqrt{k}/\varepsilon))$.  
For the second bad event, recall that $m^* \sim \Geom(p)$ with $p\ge \sqrt{k}/(\varepsilon n)$.  
By standard geometric tail bound:
\[
\Pr[m^* > t] = (1-p)^t \le \exp(-pt).
\]
Substituting $t=\tfrac13 c_1 \sqrt{k}\varepsilon n$ and $p \ge \sqrt{k}/(\varepsilon n)$ yields
\[
\Pr\!\left[m^* > \tfrac13 c_1 \sqrt{k}\varepsilon n\right]
\;\le\; \exp\!\left(- \tfrac13 c_1 \sqrt{k}\varepsilon n \cdot \tfrac{\sqrt{k}}{\varepsilon n}\right)
\;=\; \exp\!\left(- \tfrac13 c_1 k\right).
\]
Thus the probability of a bad transcript is at most $e^{-\Omega(k)}$.

Conditioned on a good transcript, we apply \cref{lem:attack-efficacy}.  
From \eqref{eq:highprob} we obtain
\[
\Pr\!\left[\hat n - n \;\ge\; \tfrac23 c_1 \sqrt{k}\varepsilon n\right] \;\ge\; 1-e^{-c_2 k}.
\]
Since good transcripts occur with probability $1-e^{-\Omega(k)}$, the same lower bound holds unconditionally (with adjusted constants $c_1,c_2>0$).  
This proves the theorem.
\end{proof}

\section{Robust Distributed Counting Protocol}\label{sec:ourRobustProtocol}

\textsf{Robust}, our robust distributed counting protocol, is described in \cref{alg:robustdc}.
The site-side protocol of \textsf{Robust} is very simple and similar to that of \textsf{HYZ12}: The site stores a local copy of a transmission probability $p$ that is updated by server broadcasts.
The site also tracks its local event count as $n_i$. 
When an event arrives, $n_i$ is incremented. With probability $p$, the site then sends a \textsf{ReportSample}$()$ message to the server. When a \textsf{CountRequest}$()$ request arrives from the server, the site sends the server its current count $n_i$.

Server-side, the protocol operates in rounds: In the beginning of each round the server broadcasts to all sites a \textsf{CountRequest}$()$, which requests their exact counts, and each site $i$ reports $n_i$, which the server stores as $\bar{n}_i$. Therefore, in the beginning of the round the server has the exact event count $\bar{n} \gets \sum_{i\in [k]}\bar{n}_i$.
The server then updates the 
transmission probability $p$, based on the updated $\bar{n}$, and broadcasts it to the sites (which update their local copy). The transmission probability remains fixed throughout the round.
The server also maintains a counter $B$ of the number of \textsf{ReportSample}$()$ messages received from sites during the round. The published estimate for the total number of events is $\hat n := \bar{n}+B/p$, and it is updated when either $B$ or $\bar n$ were updated. The round ends when the $k$th report is received and a new round is started.
\vspace{5pt}

\begin{algorithm2e}[H]
\DontPrintSemicolon
\caption{\textsf{Robust} Distributed Counting Protocol\label{alg:robustdc}}
\KwIn{number of sites $k$, accuracy $\varepsilon$, factor $c\geq 1$}
\KwOut{Running estimate $\hat n$ of the number of events $n$.}
\BlankLine
\tcp{\textbf{Site protocol for sites $i\in[k]$:}}
\textbf{Site $i$ state:} $p\gets 1$ (updated via broadcasts by server), $n_i\gets 0$ (local event counter).\;
\BlankLine
\textbf{Event arrival at site $i$:}\tcp*{runs independently at each site}
\Indp
$n_i\gets n_i+1$\;
Sample $X\sim \Bernoulli(p)$.\;
\lIf{$X=1$}{send \textsf{ReportSample}$()$ to server\tcp*{report that an event was sampled}}
\Indm
\textbf{Upon site $i$ receiving \textsf{Broadcast}$(p)$:}\quad update local $p$.\;
\textbf{Upon site $i$ receiving a \textsf{CountRequest}$()$:}\quad send \textsf{ReportCount}$(i,n_i)$ to server.\;
\BlankLine
\tcp{\textbf{Server protocol:}}
\textbf{Server state:}  Counter of messages in current round $B\gets 0$; Number of messages per round $L\gets k$; for each site $i\in [k]$:  $\bar n_i \gets 0$ (reported $n_i$ at the beginning of current round); current transmit prob $p \gets 1$.\;  
Running estimate: $\hat{n} \gets B/p + \bar{n}$; where $\bar{n} := \sum_{i\in[k]} \bar{n}_i$.
\BlankLine
\textbf{Upon server receiving \textsf{ReportSample}$()$:}\;
\Indp
$B\gets B+1$\;
\If(\tcp*[f]{close round, open new round}){$B=k$}{
Broadcast \textsf{CountRequest}$()$.\;
Collect $\bar{n}_i \gets$ \textsf{ReportCount}$(i,n_i)$ from sites $i\in[k]$\;
Update $p\gets \min\left\{1,c\frac{\sqrt{k}}{\eps \bar{n}}\right\}$; \,\textsf{Broadcast}$(p)$ 
\;$B\gets 0$

}

\Indm
\end{algorithm2e}

\subsection{Properties of the Robust Protocol}

The properties of the protocol are stated with respect to an adversary (event stream) that is \emph{white-box adaptive}.
Let $\hat n(n)$ be the server’s estimate after the $n$-th event and let
$M_N$ denote the total number of \emph{messages} sent by the time the $N$-th event occurs.
All probabilities are over the protocol’s internal randomness.

\begin{theorem}[\textsf{Robust}: per-event $(\eps,\delta)$ accuracy]\label{thm:RobustProtocol}
Fix $\delta\in(0,1)$. There exists a constant $c=c(\delta)>0$ such that, for any event stream and every $n\ge1$,
\[
\Pr\big[\,|\hat n(n)-n|>\varepsilon n\,\big]\le\delta.
\]
Moreover,
\[
\mathbb{E}[M_N]
\;=\;
O\!\left(\frac{k\,\log N}{\log\!\left(1+\frac{\sqrt{k}\,\varepsilon}{c}\right)}\right)
\;=\;
\begin{cases}
O\!\left(\dfrac{\sqrt{k}}{\varepsilon}\,\log N\right), & \varepsilon \le c/\sqrt{k},\\[6pt]
O\!\big(k\log N\big), & \varepsilon> c/\sqrt{k}
.
\end{cases}
\]
\end{theorem}

\begin{theorem}[\textsf{Robust}: uniform $(\eps,\delta)$ accuracy]\label{thm:RobustProtocolforall}
For any $\varepsilon>0$, $\delta\in(0,1)$, and  $N$, set
\[
c \;=\; O\!\left(\max\left\{\sqrt{A},\,\frac{A}{\sqrt{k}}\right\}\right),
\qquad
A \;:=\; \log\!\frac{1}{\delta}\;+\;\log\log N\;+\;\max\!\left\{0,\,\log\!\frac{1}{\sqrt{k}\,\varepsilon}\right\}.
\]
Then, for any event stream of length at most $N$, with probability at least $1-\delta$,
\[
\max_{n\in[N]}\frac{|\hat n(n)-n|}{n}\;\le\;\varepsilon,
\qquad
\text{and}\qquad
M_N \;=\; O\!\left(\frac{k\,\log N}{\log\!\left(1+\frac{\sqrt{k}\,\varepsilon}{c}\right)} \;+\; k\log\!\frac{1}{\delta}\right).
\]
\end{theorem}

\section{Analysis of the Robust Protocol}\label{sec:analysisRobust}

In this section we prove \cref{thm:RobustProtocol} and \cref{thm:RobustProtocolforall}.
In \cref{sec:robust} we argue that our protocol is robust, in \cref{sec:comm} we analyze the communication cost, 
and in \cref{sec:accuracy} we analyze the accuracy. 
 
\subsection{Robustness to an adaptive adversary}\label{sec:robust}
Robustness of the protocol is immediate.
An adaptive adversary in a white-box setting has the freedom to select which site to inject the next event to, based on the server and stations (historical and current) state and the transcript so far.
Observe, however, that the actions and effective state of the server at any given time are oblivious to the particulars of how the events are assigned to sites, and only depend on the total number of events that occurred. In particular, when syncing event counts at the beginning of each round, the server only uses $\bar{n}$ which is the total exact number of events. Likewise, the number and time of \textsf{ReportSample}$()$ messages, counted by the server in $B$, do not depend on the site at which the event had occurred.  The computed probability $p$, the termination and restart time of rounds, and the running estimate only depends on these two parameters $\bar n$ and $B$.

\subsection{Communication Analysis} \label{sec:comm}

We bound the number of rounds that the protocol performs on the first $N$ events.

\begin{lemma}[Number of rounds for $N$ events]\label{lemma:RforN}
Let $R_N$ be the number of rounds the protocol performs to process the first $N$ events. Then
\[
\mathbb{E}[R_N] \;=\; \Theta\!\left(\frac{\log N}{\log\!\left(1+ \frac{\sqrt{k}\,\varepsilon}{c} \right)}\right).
\]
Moreover, for any $\delta \in (0,1)$, with probability at least $1-\delta$,
\[
R_N \;=\; O\!\left(\frac{\log N}{\log\!\left(1+ \frac{\sqrt{k}\,\varepsilon}{c} \right)} \;+\; \log\frac{1}{\delta} \right),
\]
where the implicit constant in the $O(\cdot)$ notation is absolute and independent of $\delta,\varepsilon,k,N$.
\end{lemma}

\begin{proof}
The number of events in one round is
\[
S \;=\; \sum_{i=1}^k X_i,
\qquad X_i \stackrel{\text{i.i.d.}}{\sim} \Geom(p),
\quad p \;=\; \min\!\left\{1,\frac{c\sqrt{k}}{\varepsilon\,\bar n}\right\},
\]
where $\Geom(p)$ is the geometric distribution on $\{1,2,\ldots\}$. Hence
$\E[S]=k/p \;\ge\; \max\{k,(\varepsilon\sqrt{k}/c)\,\bar n\}$.

Call a round \emph{big} if
\[
S \;\ge\; \tfrac12\,\mathbb{E} [S] \;\ge\; \frac{\varepsilon\sqrt{k}}{2c}\,\bar n .
\]
For $X\sim\Geom(p)$ we have $\Pr(X\ge 1/(2p))\ge 1/2$; by standard properties of sums of independent geometrics (negative binomial concentration), this implies
\[
\Pr\!\left(S \;\ge\; \tfrac12\,\mathbb{E}[S]\right)\;\ge\;\tfrac12 .
\]
Thus each round is big with probability at least $1/2$, independently of other rounds.

Let
\[
s^\star \;=\; \left\lceil \log_{\,1+\frac{\varepsilon\sqrt{k}}{2c}} N \right\rceil .
\]
After $s^\star$ big rounds we must have processed at least $N$ events. Let
$T$ be the number of rounds needed to accrue $s^\star$ big rounds. Then
$T$ is stochastically dominated by the number of $\Bernoulli(1/2)$
trials needed to see $s^\star$ successes, so $\mathbb{E}[T]\le 2s^\star$, and by a standard Chernoff lower-tail bound, for all $\delta\in(0,1)$,
\[
T \;\le\; 2\!\left(s^\star + 2\sqrt{s^\star \ln\!\tfrac{1}{\delta}} + 4\ln\!\tfrac{1}{\delta}\right)
\quad\text{with probability at least }1-\delta .
\]
Since $R_N\le T$, the high-probability bound follows. Moreover,
\[
\mathbb{E}[R_N]\;\le\;\mathbb{E}[T]\;\le\;2s^\star
\;=\; \Theta\!\left(\frac{\log N}{\log\!\left(1+\frac{\sqrt{k}\,\varepsilon}{c}\right)}\right),
\]
as replacing the base $1+\varepsilon\sqrt{k}/(2c)$ by $1+\varepsilon\sqrt{k}/c$ only changes the expression by an absolute constant factor. This completes the proof.
\end{proof}

\subsubsection{Communication on the first $N$ events}

In each round the protocol sends $O(k)$ messages. We combine with \cref{lemma:RforN} to obtain 
\[
\mathbb{E}[M_N] \;=\; O\!\left(k\,\frac{\log N}{\log\!\left(1+\frac{\sqrt{k}\,\eps}{c}\right)}\right).
\]
To establish the message bound in \cref{thm:RobustProtocol}, take $c=\Theta(1)$; then 
\(\log\!\bigl(1+\tfrac{\sqrt{k}\eps}{c}\bigr)=\Theta(\min\{1,\sqrt{k}\eps\})\).

Similarly from \cref{lemma:RforN}  we obtain that for any $\delta\in(0,1)$, with probability at least $1-\delta$,
\[
M_N \;=\; O\!\left(k\,\frac{\log N}{\log\!\left(1+\frac{\sqrt{k}\,\eps}{c}\right)} \;+\; k\log\frac{1}{\delta}\right).
\]
To establish the message bound in 
\cref{thm:RobustProtocolforall}, we substitute
\(c=O\!\bigl(\max\{\sqrt{A},\,A/\sqrt{k}\}\bigr)\) (with \(A\) as defined there).

\subsubsection{Communication bounds in bits}\label{sec:bits}
Following prior work, we measure communication in terms of \emph{messages}.  We present a lightly modified version of our protocol (\cref{alg:robustdc}) in which each message has size $O(\log(k/\eps))$ bits, improving over the $O(\log N)$-bit messages required in a vanilla implementation.

\begin{lemma}[Communication bound in bits]
There is a (lightly modified) version of \cref{alg:robustdc} with the same accuracy guarantees whose per-round communication is
\[
O\!\left(
  k \;+\;
  \min\!\left\{\,k,\ \frac{1}{\varepsilon^{2}}\log\frac{1}{\delta}\,\right\}
  \cdot \log\frac{k}{\varepsilon}
\right)
\ \ \text{bits.}
\]
\end{lemma}

\begin{proof}
The vanilla description of the protocol (\cref{alg:robustdc}) uses messages that are event-triggered from sites to the server. These messages are of size $O(1)$. Additionally, messages of potentially larger sizes are used for (i) syncing the updated transmission probability and (ii) event counts in the beginning of each round. We describe the modifications needed to condense this communication.

\paragraph{Transmission probability syncing: }
Our analysis goes through when the server quantizes the probabilities to powers of $2$ (that is, sets
$p \gets 2^{\min\{0,\lceil \log_2 \frac{c\sqrt{k}}{\eps \bar n} \rceil\}}$. In each broadcast, the server then only sends the increase to the previous value of the negated exponent. Therefore, the total bits broadcasted, over all rounds, are 
\[O(\log_2 (\eps N/(c\sqrt{k}))=O(\log_2 N).\] Since each broadcast is to $k$ sites, this amounts to $k\log_2 N$ total bits communicated to facilitate the updates of $p$.  Since $\E[R_N] = \Omega(\log N)$, syncing $p$ does not add asymptotically to the total communication in bits.

\paragraph{Total event count syncing: } 
At the end of each round, the server collects the event counts $\bar n_i$ from the sites $i\in [k]$. 
Observe that the server only needs to obtain an \emph{estimate} of the total count $\bar n$ that is within a relative error of $\eps/2$ (and then provide $(\eps/2)$ accuracy guarantees). 

Obtaining an estimated count is is an instance of 
\emph{distributed approximate counting} (without tracking), with the added benefit that the sites (and the server) share a high probability constant-factor proxy $\bar n'$ for $\bar n$
(e.g., $\bar n'=\Theta(c\sqrt{k}/(\varepsilon p) + k/p)$).

It therefore suffices that each site broadcasts its value to within an additive error of $O(\max\{1,\frac{\eps \bar n'}{k}\})$. When $n_i = O(\bar n')$, this requires 
$O(\log(k/\eps)$ bits per message and $O(k\log(k/\eps))$ bits per round.\footnote{This can be converted to a high probability or in-expectation bound to cover the case where the coarse estimate $\bar n'$ is off.}

In the regime where $\eps > 1/\sqrt{k}$, we can use weighted sampling, where sites only share their approximate count probabilistically. In this case, only $\Theta\!\left(\dfrac{1}{\varepsilon^{2}} \log \dfrac{1}{\delta}\right)$ sites send messages to obtain $(\eps,\delta)$-accurate $\bar n$. 

\end{proof}

\subsection{Accuracy Analysis} \label{sec:accuracy}

We first bound the probability that the maximum relative error over all the event times in a single round exceeds $\eps$:
\begin{lemma} [Tail bound on the maximum relative error in a round] \label{lemma:perrounderrorprob}
    Consider the random variable that is a single round of the protocol that is started at some event count $n_0$. Let $R$ be the number of events in the round and $\hat{n}$ the estimate of the event count when the event count is $n$. Then
    \[\Pr\left[\max_{n\in [t,t+R-1]}\frac{\vert n-\hat{n} \vert}{n} > \eps\right] \leq 2\exp\!\left(-\min\!\left\{\frac{c^2 }{8},\ \frac{c\sqrt{k}}{4}\right\}\right).\]
\end{lemma}
\begin{proof}

We use the following fact  (see \cref{sec:tails} for a proof)
\begin{claim} [Bound on the maximal partial-sum deviation \cite{boucheron2013concentration}]
    Let $X_1,\ldots,X_r \stackrel{\mathrm{iid}}{\sim}\mathrm{Geom}(p)$. 
Define $S_i=\sum_{j=1}^i X_j$.
Then for all \(t>0\),
\[
\Pr\!\left(\max_{1\le i\le r}\Big|S_i-\tfrac{i}{p}\Big|\ge t\right)
~\le~
2\exp\!\left(-\min\!\left\{\frac{t^2 p^2}{8r},\ \frac{tp}{4}\right\}\right).
\]
\end{claim}

We apply the claim with $t=\eps \bar{n}$, $r=k$ and $p=c\sqrt{k}/(\eps \bar{n})$ and thus $pt= c\sqrt{k}$ and obtain that the probability of a deviation that exceeds $\eps \bar{n}$ at any of the message times (update points) in the round is at most 
\[
2\exp\!\left(-\min\!\left\{\frac{c^2 }{8},\ \frac{c\sqrt{k}}{4}\right\}\right).
\]

Finally, observe that this bound on the maximum deviation at \emph{message times} translates to a bound of $\eps \bar{n} +1/p = (1+1/(c\sqrt{k}))\eps\bar{n}$ on the maximum deviation between the actual and estimated counts over \emph{all events} in the round, since the estimate is updated in increments of $1/p$.\footnote{technically we should apply this with the slightly smaller $\eps' = \eps/(1+1/(c\sqrt{k}))$ but we skip that to reduce clutter}
\end{proof}

\subsubsection{$(\eps,\delta)$-accuracy on the first $N$ events}
\begin{proof}[Proof of \cref{thm:RobustProtocolforall}; accuracy]
From \cref{lemma:perrounderrorprob}, using
\[c= O\left(\max\left\{\sqrt{\log\left(\frac{R}{\delta}\right)}, \frac{\log\left(\frac{R}{\delta}\right)}{\sqrt{k}}\right\}\right) ,\]  with probability at least $1-\delta$, the maximum relative error  at all events in the first $R$ rounds is at most $\eps$.
This follows from a per-round bound on the error probability of $\delta/R$ and a union bound on the $R$ rounds.
To obtain the bound for the first $N$ events, 
we substitute for $R$ the upper tail bound on $R_N$, the number of rounds in $N$ events, given in \cref{lemma:RforN}.
\end{proof}

\subsubsection{Per-event-index $(\eps,\delta)$-accuracy}

From \cref{lemma:perrounderrorprob}, using  
\begin{equation}\label{RroundsAcc:eq}
     c=O\left(\max\left\{\sqrt{\log\left(\frac{1}{\delta}\right)} , \frac{\log\left(\frac{1}{\delta}\right) }{\sqrt{k} }\right\} \right)         \, ,    
\end{equation}
we obtain that the protocol provides \emph{per-round $(\eps,\delta)$ accuracy}, that is, for each round (independently of history), 
with probability at least $1-\delta$, the maximum relative error of the count at all events that fall in the round is at most $\eps$.

To
establish per-event-index $(\eps,\delta)$ accuracy, 
we relate it to the 
per-round accuracy via the following (proof in \cref{sec:perevent}):
\begin{lemma} [Per-round to Per-event success bound]\label{lem:perR2perE}
For any $\pi>0$, with $\delta(\pi) = O(\pi/\log(1/\pi))$, the protocol with per-round $(\eps,\delta(\pi))$ accuracy 
provides $(\eps,\pi)$ per-event-index accuracy.
\end{lemma}

\begin{proof}[Proof of \cref{thm:RobustProtocol}; accuracy]
We apply the protocol with
\[c = O\left(\max\left\{\sqrt{\log(1/\delta(\pi))} , \frac{\log(1/\delta(\pi)) }{\sqrt{k} }\right\} \right) .
\]
Per \cref{RroundsAcc:eq} this gives per-round $(\eps,\delta(\pi))$ accuracy and per \cref{lem:perR2perE} this translates to $(\eps,\pi)$ per-event-index accuracy.

Now note that
\[\max\left\{\sqrt{\log(1/\delta(\pi))} , \frac{\log(1/\delta(\pi)) }{\sqrt{k} }\right\} = O\left(\max\left\{\sqrt{\log(1/\pi)}, \frac{\log(1/\pi) }{\sqrt{k} }\right\} \right) .\] Therefore we can express $c$ as stated in terms of $\pi$ (with a different hidden constant). 
\end{proof}

\section{Per-round to per-event-index accuracy}\label{sec:perevent}

In this section we prove \cref{lem:perR2perE}. We derive a bound on the maximum per-event-index failure probability from the per-round failure probability $\delta$. For the purpose of this analysis we introduce a simple
stochastic renewal process which captures 
the relevant components of the protocol's execution.

\paragraph{Renewal process}.
Fix integers $k\ge1$ and a constant $C>0$. For $n\geq 1$ let
\[
S_n \;=\; \sum_{i=1}^k X_i, \qquad X_i \stackrel{\text{iid}}{\sim} \Geom\!\big(p(n)\big),
\quad p(n):=\min\{1,C/n\}.
\]
 
Write
\[
\mu_n:=\E[S_n]=\frac{k}{p(n)}=\frac{k}{C}\,n,\qquad
\nu_n^2:=\Var(S_n)=\frac{k\,(1-p(n))}{p(n)^2} .
\]

Partition the outcomes of $S_n$ into $A_n$ (good) and $B_n$ (bad) with
\[
\Pr(B_n)\le \delta\quad\text{and}\quad \Pr(A_n)=1-\Pr(B_n)\ge 1-\delta,
\]
where the partition depends only on $n$.

Our renewal process starts with $N_0=1$ and uses $N_{j+1}=N_j+L_j$, where
$L_j\sim S_{N_j}\mid A_{N_j}$ if the outcome falls in $A_{N_j}$,
and $L_j\sim S_{N_j}\mid B_{N_j}$ otherwise.
Let $T_j\in\{A,B\}$ be the label of $L_j$.
For a time $t\ge1$, let $J(t)=\max\{j:N_j\le t\}$.
We say that \emph{$t$ is bad} if $T_{J(t)}=B$.

\begin{remark}[Mapping to our protocol]
    $S_n$ samples the length of a round that starts at event-index $n$. The partitioning to label $B_n$ corresponds to a round being bad (not meeting the maximum accuracy requirement) and $A_n$ to the round meeting the requirement. 

    A time $t$ in the process corresponds to an event-index $t$ in the protocol. When $t$ is good, the round is good and the accuracy requirement is met. When $t$ is bad, the round is bad and the accuracy requirement may not have been met. We aim to bound the probability that $t$ is bad.
\end{remark}

\paragraph{Tail bounds.}
We use the standard sub-exponential tail (Bernstein/Bennett form): for all $x\ge0$,
\begin{equation}\label{eq:bernstein}
\Pr\!\big(S_n-\mu_n\ge x\big)\ \le\ \exp\!\left(-\tfrac12\,
\min\!\left\{\frac{x^2}{\nu_n^2},\ x\,p(n)\right\}\right).
\end{equation}

Inverting \eqref{eq:bernstein} yields, for every $n$, the 
quantile bound 
\begin{equation}\label{eq:qdelta}
q_\delta(n)\ \le\ \mu_n\ +\ 2\max\!\Big\{\nu_n\sqrt{\log\frac{1}{\delta}},\ \frac{\log\frac{1}{\delta}}{p(n)}\Big\},
\end{equation}
where $q_\delta(n)$ denotes the $(1-\delta)$-quantile of $S_n$.

\bigskip

\begin{claim}[Quantile-to-gap bound for $t_0$]\label{claim:t0-gap}
Fix $t\in\mathbb{N}$ and define $t_0<t$ to be the largest $t_0\geq 0$ such that
\[
\Pr\!\big[S_{t_0}>t-t_0\big]\ \le\ \delta.
\]
If $t_0$ is not defined we take $t_0=0$. Then
\[
\quad
\frac{t-t_0}{\mu_{t_0}}
\ \le\
1\;+\;4\,\max\!\Big\{\sqrt{\tfrac{\log\frac{1}{\delta}}{k}},\ \tfrac{\log\frac{1}{\delta}}{k}\Big\}\;+\;\frac{2}{t_0} = O(1+\tfrac{\log\frac{1}{\delta}}{k}) \,.
\quad
\]
\end{claim}
\begin{proof}
Let $s:=t-t_0$ and note that by its definition
$s\ \le\ 1 + q_\delta(t_0)$. 
Substitute \eqref{eq:qdelta} and divide 
by $\mu_{t_0}$ to obtain the claim.
\end{proof}

\bigskip

We try to provide an informal overview of our approach to bounding the probability that $t$ is bad: Observe that the probability that time $t$ is bad is bounded by $\delta$ times the expected number of round starts before $t$. The issue in using this, however, is that for a large $t$, this expected number of round starts in $[0,t]$ can be large. In this case, our strategy is to look at a more restricted window $[t_0,t]$. We choose $t_0(t)$ so that the probability that no round starts in $[t_0,t]$ is at most $\delta$. We then pay this $\delta$ towards our bound and only consider rounds starting in this window. 

\begin{lemma}[Per-\texorpdfstring{$t$}{t} bad-probability via the window \texorpdfstring{$[t_0,t]$}{[t0,t]}]
\label{lem:per-t-window}
With $t_0$ as in \cref{claim:t0-gap}, let $M_t:=\sum_j \mathbf{1}\{N_j\in[t_0,t]\}$ be the
number of round starts in $[t_0,t]$, where $(N_j)$ are the renewal start positions generated by
the process. Then
\[
\quad
\Pr(\text{$t$ lies in a bad round})\ \le\ \delta\;+\;\delta\,\E[M_t]
\ \le\ \delta\Big(2+\frac{t-t_0}{\mu_{t_0}}\Big)\,.
\quad
\]
\end{lemma}

\begin{proof}
\textbf{Show that the probability that no round starts in $[t_0,t]$ is at most $\delta$:}
Define $f(n):=\Pr(S_n>t-n)$. Since $p(n)=\min\{1,C/n\}$ is non increasing in $n$, the law of $S_n$ is stochastically
increasing in $n$, while $t-n$ decreases; hence $f$ is non-decreasing. Let $W:=N_{J(t)}$ be the start of the round that covers $t$. The event
``no start in $[t_0,t]$'' is exactly $\{W\le t_0\ \&\ S_W>t-W\}$, so
\[
\Pr(\text{no start in }[t_0,t])=\E\big[f(W)\big]\ \le\ f(t_0)\ \le\ \delta,
\]
where the expectation is over the randomness of the process generating $W$.
\bigskip

\textbf{Bound bad probability of $t$ by union over starts in $[t_0,t]$:}
On $\{M_t\ge1\}$, exactly one start $n\in[t_0,t]$ has its interval covering $t$.
Let $w(n):=\Pr(\exists j:\ N_j=n)$ denote the (unconditional) probability of a start at $n$.
By the Markov property at $n$ and because $\Pr(B_n)\le\delta$,
\[
\Pr(t\text{ bad},\,M_t\ge1)
\le \sum_{n=t_0}^t w(n)\,\Pr\big(B_n\big)
\le \delta \sum_{n=t_0}^t w(n)=\delta\,\E[M_t].
\]

\textbf{Bound $\E[M_t]$ by coupling to an i.i.d.\ baseline and Wald.} We show that $\E[M_t]$ is bounded from above by a similar process where round lengths are all $S_{t_0}$. 
Let $j^\star:=\min\{j:\ N_j\ge t_0\}$ and set $H:=t-N_{j^\star}\in[0,t-t_0]$.
For $n\ge t_0$ we have $S_n\succeq S_{t_0}$ (stochastic dominance since $p(n)$ is non-increasing).
Couple the post-$t_0$ interarrivals $L_{j^\star+i}\sim S_{N_{j^\star+i}}$ with i.i.d.\ $Y_i\sim S_{t_0}$
via common uniforms so that $L_{j^\star+i}\ge Y_i$ a.s.
Let $N'(h):=\max\{m\ge0:\ Y_1+\cdots+Y_m\le h\}$.
Then pathwise $\#\{j: N_j\in(N_{j^\star},t]\}\le N'(H)$, hence
\[
M_t\ \le\ 1+N'(H).
\]
Conditioning on $H$ and using Wald’s bound for nonnegative i.i.d.\ steps,
$\E[N'(H)\mid H]\le H/\mu_{t_0}$, so
\[
\E[M_t]\ \le\ 1+\E\!\left[\frac{H}{\mu_{t_0}}\right]\ \le\ 1+\frac{t-t_0}{\mu_{t_0}}.
\]

\textbf{Combining.}
The probability that $t$ is bad is decomposed to the events of ``no round started'' in the interval $[t_0,t]$ and ``at least one started.'' From the above we obtain 
\[
\Pr(t\text{ bad})\ \le\ \delta\ +\ \delta\,\E[M_t]\ \le\ \delta\ +\ \delta\Big(1+\frac{t-t_0}{\mu_{t_0}}\Big)
\ =\ \delta\Big(2+\frac{t-t_0}{\mu_{t_0}}\Big).
\]
\end{proof}

\begin{proof}[Proof of Lemma~\ref{lem:perR2perE}]
The claim follows by combining \cref{claim:t0-gap} and \cref{lem:per-t-window}: for every fixed $t$,
\[
\Pr(\text{$t$ is bad})
\ \le\  O\bigg(\delta(1+\tfrac{\log\frac{1}{\delta}}{k})
\Bigg).
\]
\end{proof}

\section{Empirical Demonstration} \label{sec:experiments}

We simulate the \textsf{HYZ12} protocol as described in \cref{alg:HYZ12}, %
and our \textsf{Robust} protocol as described in \cref{alg:robustdc}.

\smallskip
\noindent\textbf{Inputs.}
For each configuration \((k,\varepsilon)\) (with $c=1.0$ for \textsf{Robust}) we evaluate each of the protocols on the following two input streams:
(i) \emph{Uniform}, where each event is assigned to an i.i.d.\ uniformly random site in \([k]\);
and (ii) an \emph{Adaptive attack} that repeatedly injects events at the current site until
the published estimate changes, and then advances to the next site in round-robin (corresponding to our adaptive attack on the \textsf{HYZ12} protocol, as described in \cref{alg:attack}).

\smallskip
\noindent\textbf{Metrics.}
At each event index $t\le N$ we report: 
(a) cumulative communication $M_t$, 
with server broadcasts counted as \(k\) messages, (b) the ratio \(\hat{n}_t/n_t\),
and (c) the relative error \(|\hat{n}_t - n_t|/n_t\).
For each configuration we run \(r\) independent seeds and, at each event index \(t\le N\),
plot the median across runs with a $(0.05,0.95)$ quantile band.

\paragraph{Performance of \textsf{HYZ12}.}
\cref{fig:hyz12_k_eps0125_r40} 
reports representative results
for $k\in\{64,256\}$ and $\eps=0.125$. We observe that communication under uniform and attack streams is comparable, but accuracy differs: Under uniform (non-adaptive), the estimates concentrate around the true values as predicted by \cite{huang2012randomized}; under attack input (adaptive), a persistent bias emerges. The relative bias level depends on $(k,\eps)$ and remains steady as the number of events $t$ grows.

\paragraph{Performance of \textsf{Robust}.}
\cref{fig:Robust_k_eps0125_r40} reports the same metrics for \textsf{Robust}. Communication and accuracy are similar for uniform and attack streams, and the estimate remains (empirically) unbiased in both cases, with comparable estimation error. The observed variability in the band width is due to increased variability within each round as the round progresses.

Overall the simulations illustrate \textsf{HYZ12}'s vulnerability to adaptive inputs and the robustness of our protocol \textsf{Robust}. They also show that \textsf{Robust} behaves similarly across input streams regardless of which sites receive the events.

\begin{figure}[t]
  \centering
  \begin{subfigure}{.32\linewidth}
    \includegraphics[width=\linewidth]{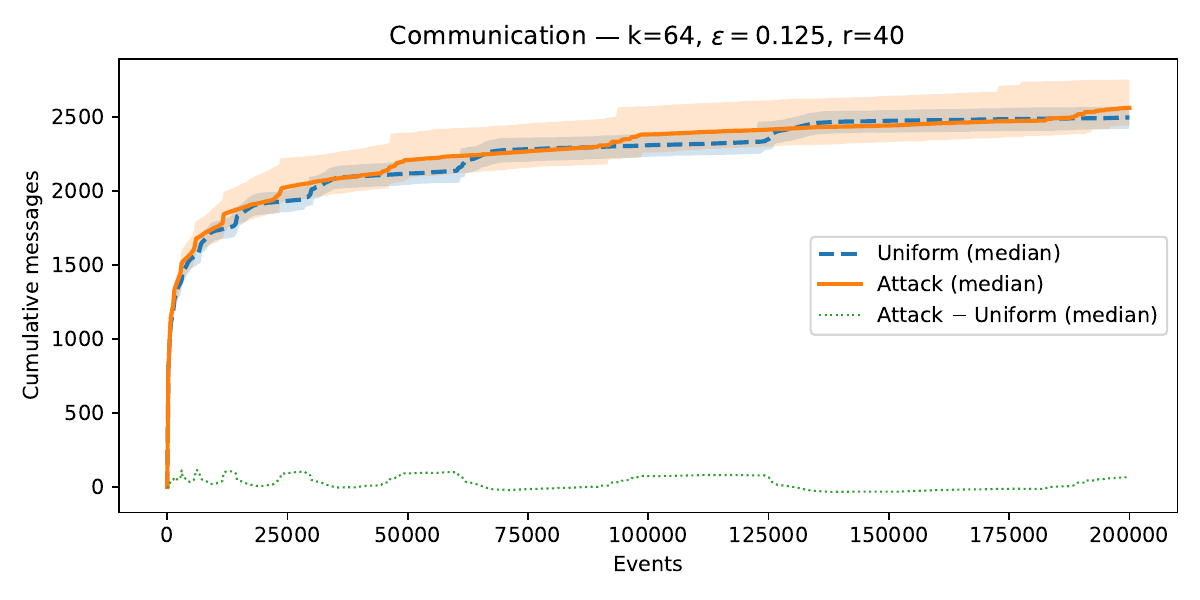}
        \includegraphics[width=\linewidth]{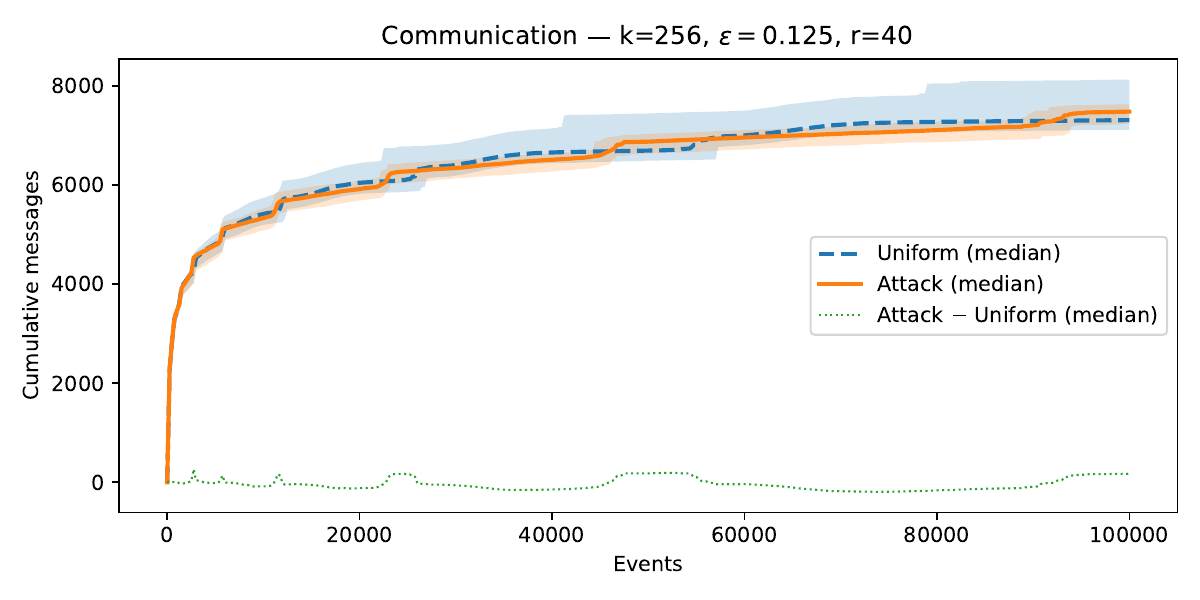}
    \caption{Communication}
  \end{subfigure}\hfill
  \begin{subfigure}{.32\linewidth}
    \includegraphics[width=\linewidth]{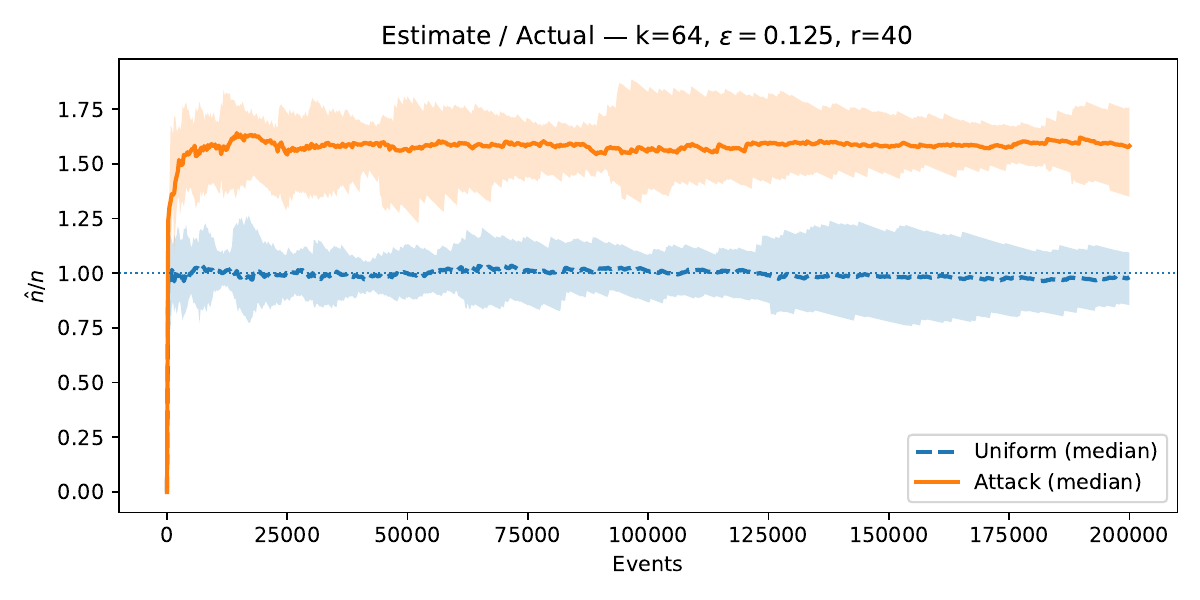}
        \includegraphics[width=\linewidth]{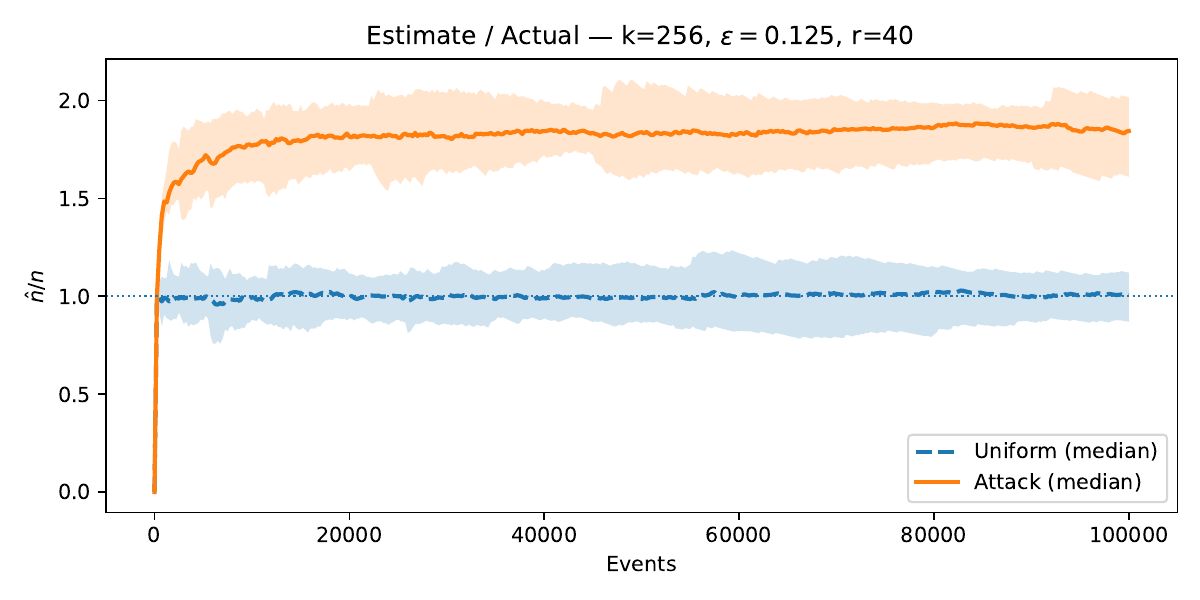}
    \caption{Estimate / actual \(\hat{n}/n\)}
  \end{subfigure}\hfill
  \begin{subfigure}{.32\linewidth}
    \includegraphics[width=\linewidth]{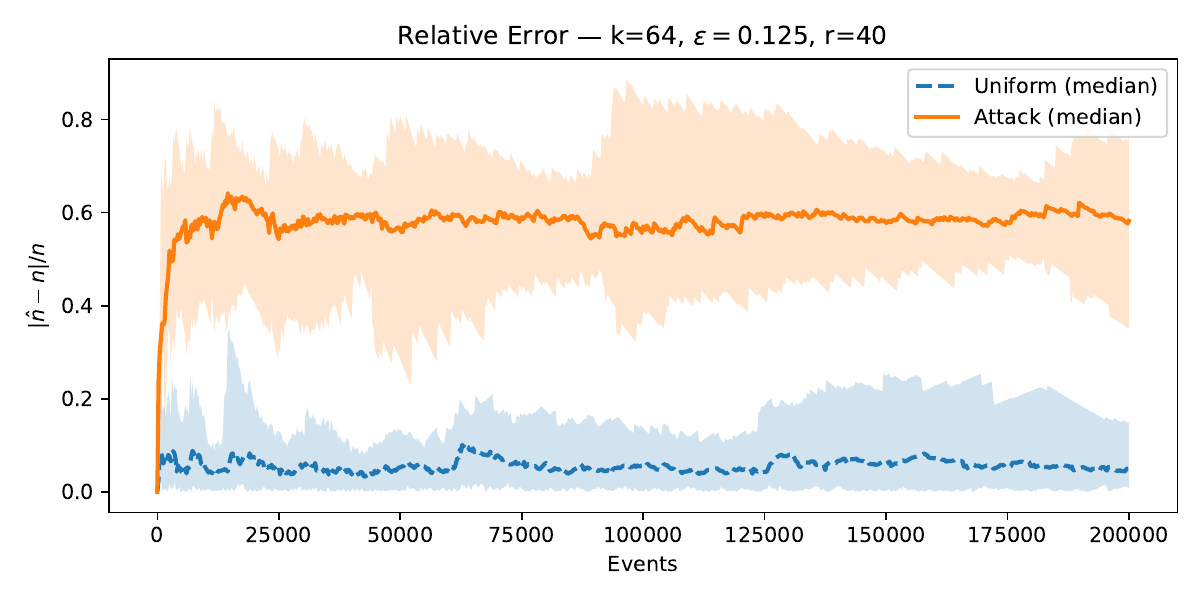}
       \includegraphics[width=\linewidth]{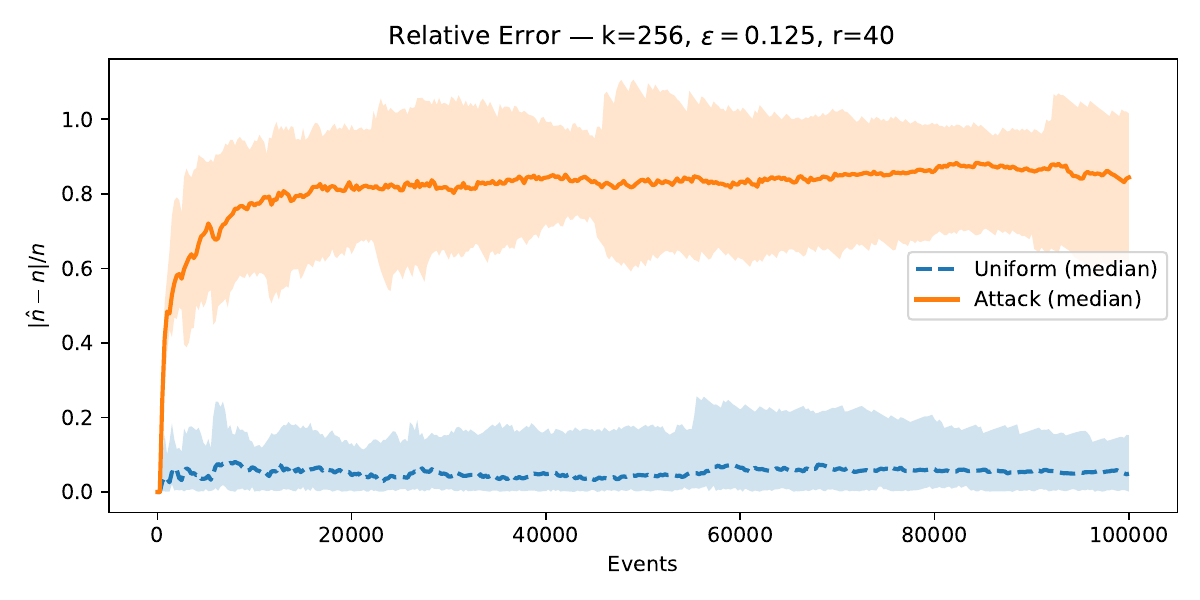}
    \caption{Relative error \(|\hat{n}-n|/n\)}
  \end{subfigure}
  \caption{\textsf{HYZ12} under uniform (dashed) vs.\ attack (solid).
  \(k=64,256\), \(\varepsilon=0.125\), \(r=40\).}
  \label{fig:hyz12_k_eps0125_r40}
\end{figure}

\begin{figure}[t]
  \centering
  \begin{subfigure}{.32\linewidth}
    \includegraphics[width=\linewidth]{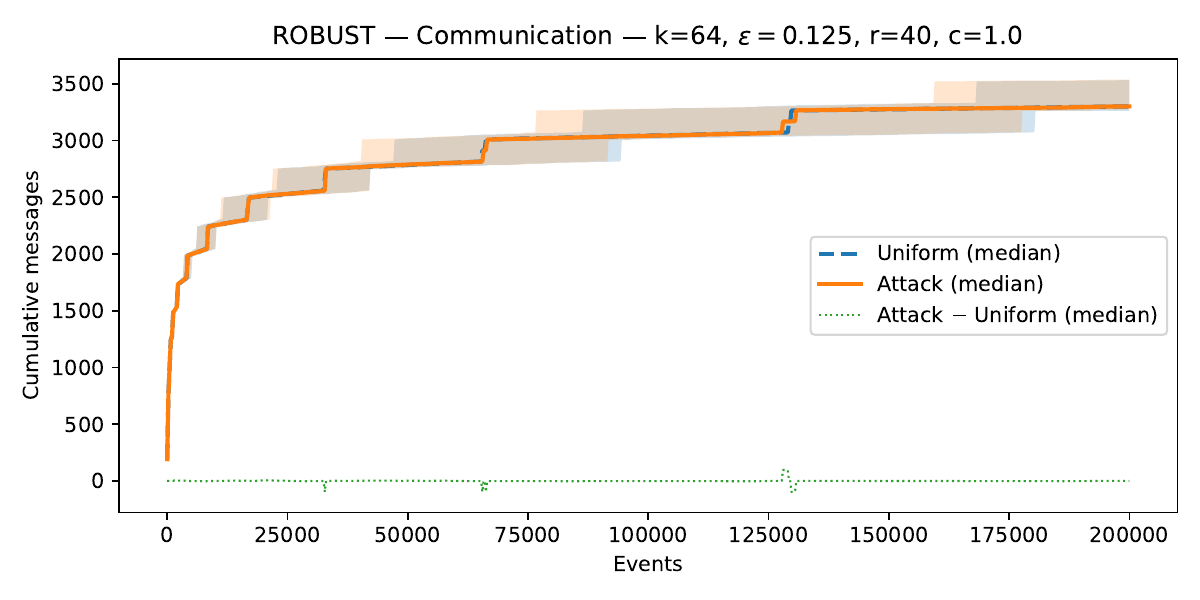}
        \includegraphics[width=\linewidth]{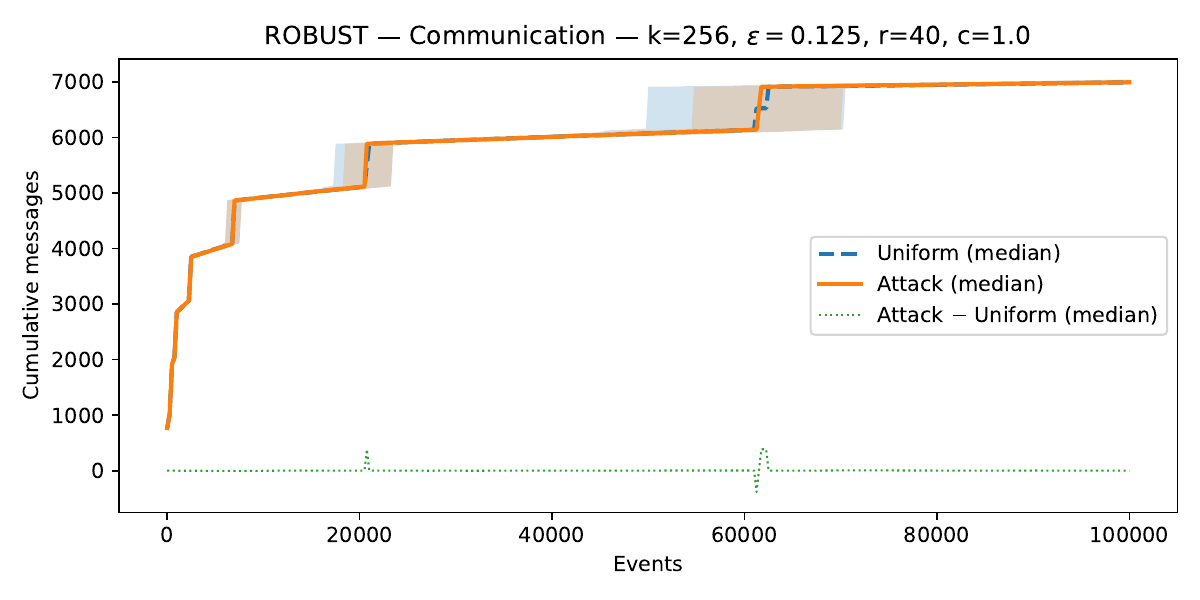}
    \caption{Communication}
  \end{subfigure}\hfill
  \begin{subfigure}{.32\linewidth}
    \includegraphics[width=\linewidth]{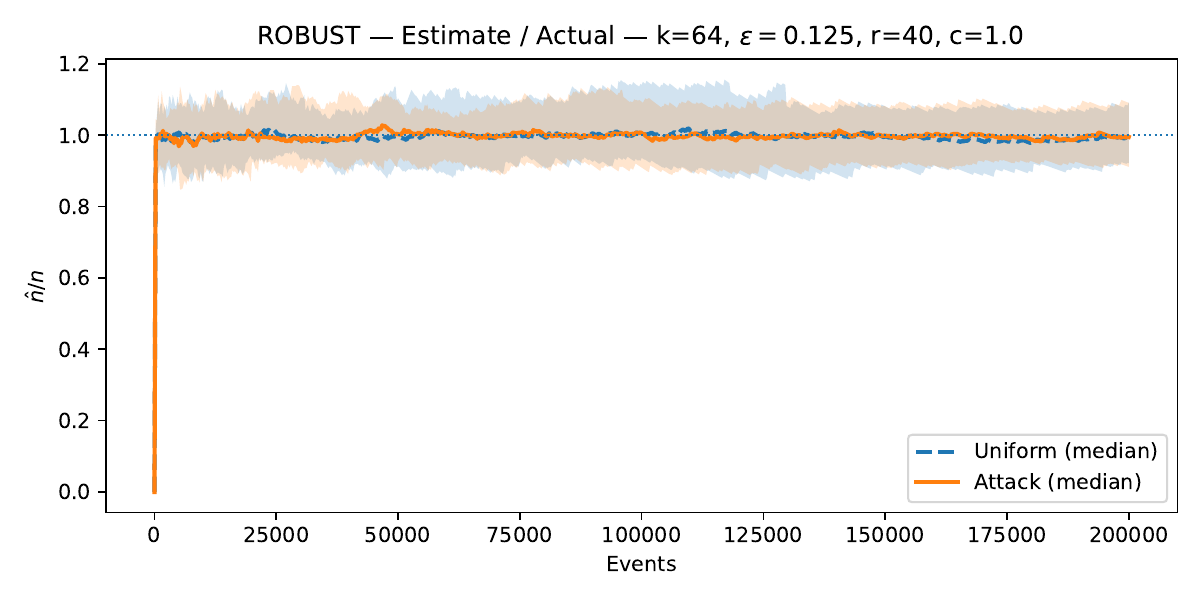}
        \includegraphics[width=\linewidth]{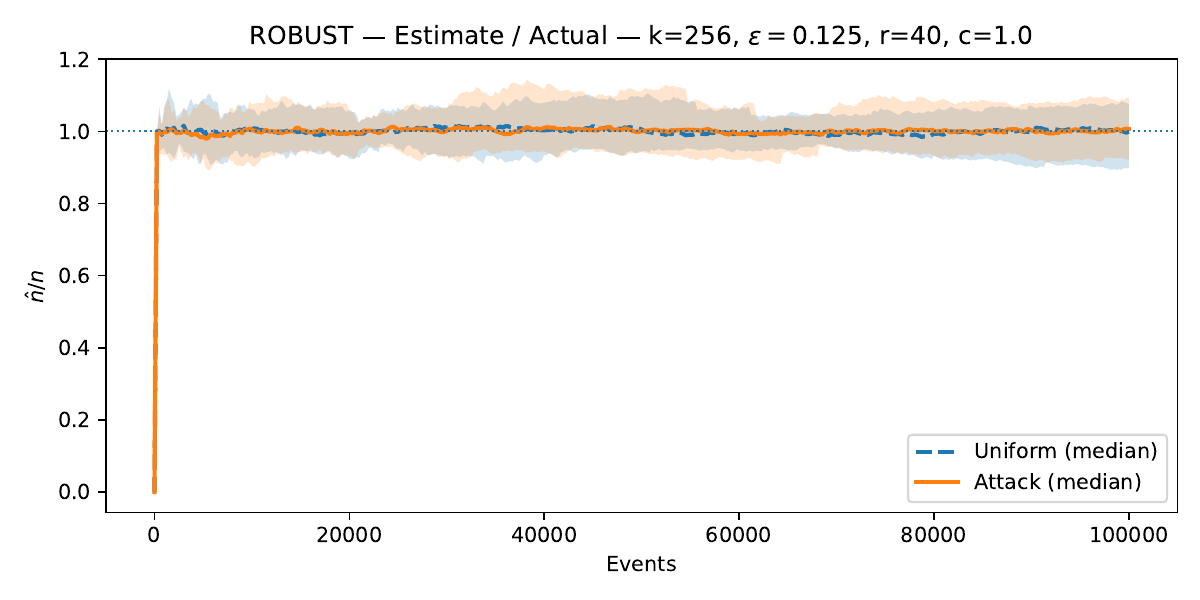}
    \caption{Estimate / actual \(\hat{n}/n\)}
  \end{subfigure}\hfill
  \begin{subfigure}{.32\linewidth}
    \includegraphics[width=\linewidth]{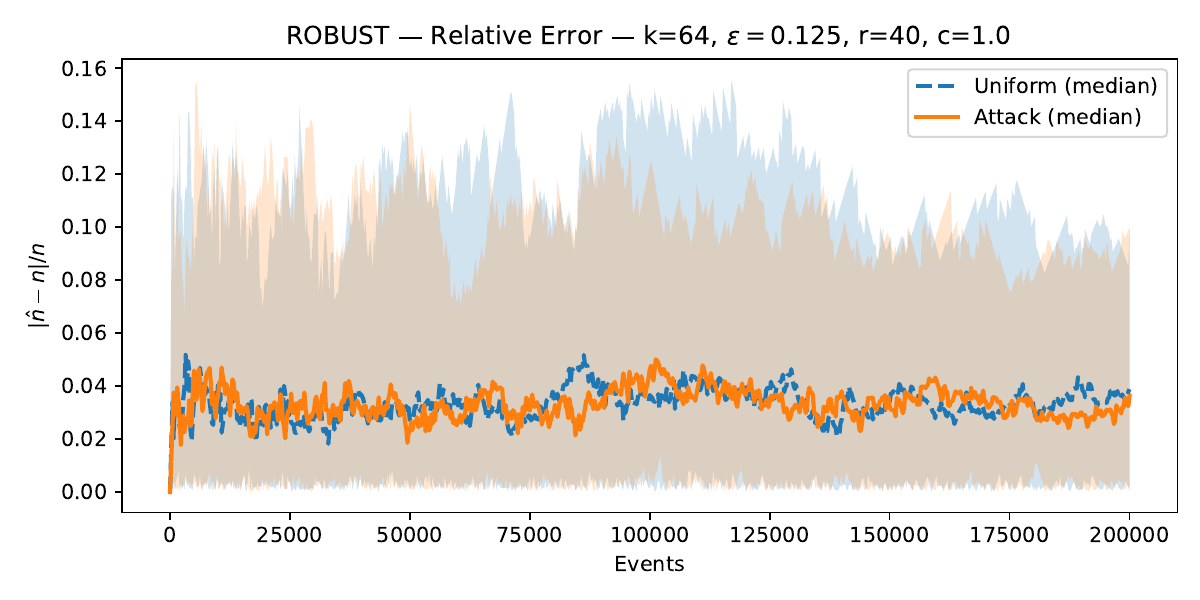}
       \includegraphics[width=\linewidth]{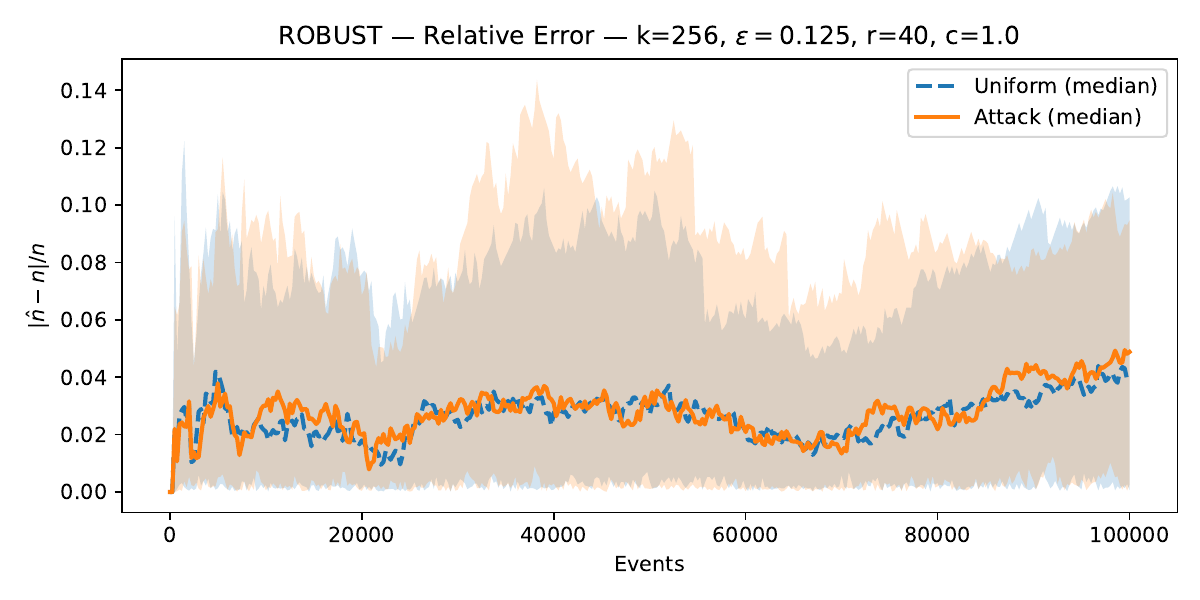}
    \caption{Relative error \(|\hat{n}-n|/n\)}
  \end{subfigure}
  \caption{\textsf{Robust} under uniform (dashed) vs.\ attack (solid).
  \(k=64,256\), \(\varepsilon=0.125\), \(r=40\).}
  \label{fig:Robust_k_eps0125_r40}
\end{figure}

\paragraph{Tradeoff experiment.}
We compare the communication--accuracy tradeoff of \textsf{HYZ12} and \textsf{Robust} for fixed $k$ and $N=10^5$ on both streams. 
We fix $k$ and sweep $\varepsilon$ over the grid $\varepsilon_i = 2^i/\sqrt{k}$ for $i\in\{-3,-2,-1,0,1,2\}$, retaining only values with $\varepsilon_i\le 1/2$ and set $c=1$ for \textsf{Robust}.

For each $(\text{protocol}, \text{stream}, \varepsilon)$ we run $r$ independent seeds and summarize each run by
\[
\text{Comm} := M_{N}\,,
\qquad
\text{Acc} := \frac{1}{N/2+1}\sum_{t=\lfloor N/2\rfloor}^{N}\frac{|\hat n_t - t|}{t}\,.
\]
Across the $r$ seeds we report the \emph{median} of \text{Comm} and \text{Acc} as the point for that configuration. 
Error bars on both axes are central $90\%$ quantile bands (5th and 95th percentiles) across seeds.

The reported results in \cref{fig:tradeoff-grid} show for each fixed $k$ two figures (one per input stream type).
Each figure shows, for both protocols, the six $(\text{Acc},\text{Comm})$ points corresponding to the $\varepsilon$ sweep, connected by lines. 
The horizontal axis is the average relative error on the interval $t\in[N/2,N]$; 
the vertical axis is the cumulative messages at $t=N$. 

We observe that for the adaptive attack inputs, \textsf{Robust} uniformly dominates; for uniform inputs, the curves are  incomparable: \textsf{HYZ12} is favored in the very low $\eps\ll 1/\sqrt{k}$ regime while \textsf{Robust} is better in the moderate $\eps$ regime. Recall that both protocols have the same asymptotic guarantees. 

\begin{figure}[t]
  \centering
  \begin{subfigure}{.48\linewidth}
    \includegraphics[width=\linewidth]{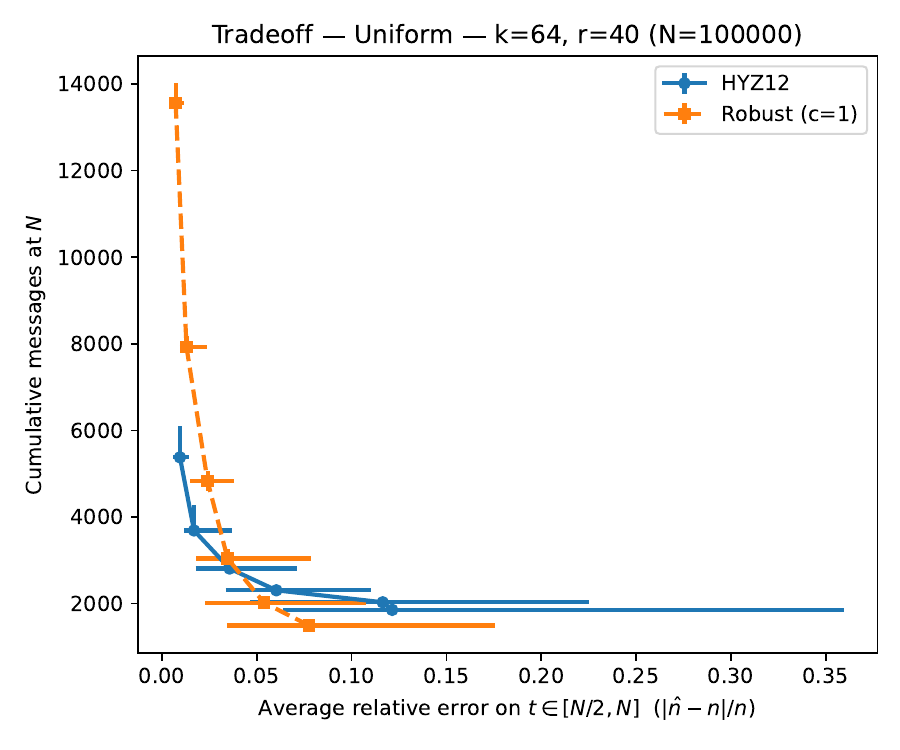}
    \caption{\(k=64\), Uniform}
  \end{subfigure}\hfill
  \begin{subfigure}{.48\linewidth}
    \includegraphics[width=\linewidth]{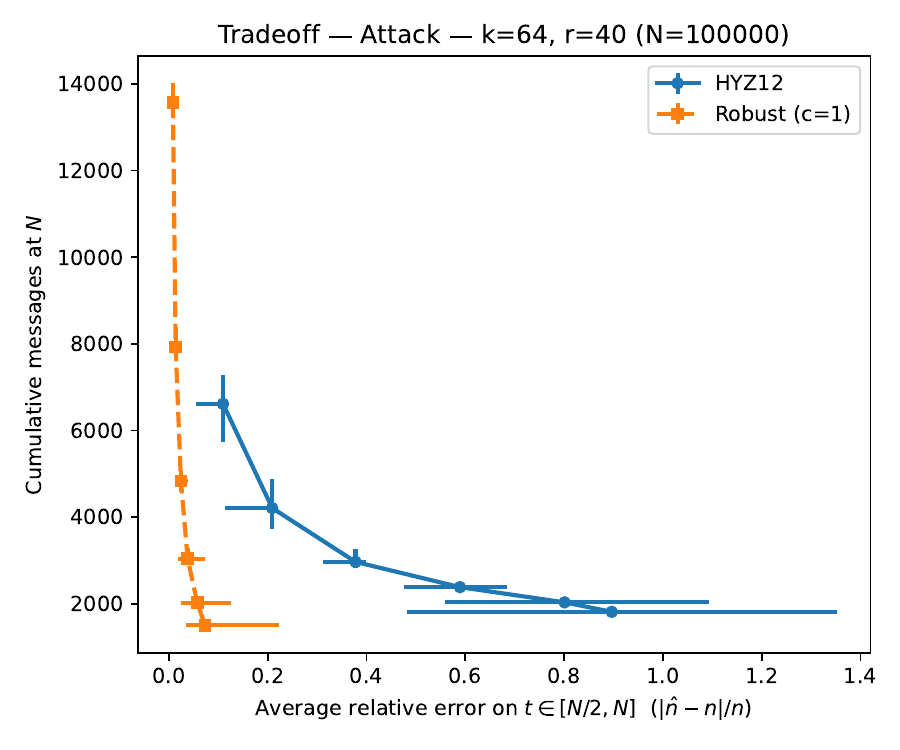}
    \caption{\(k=64\), Attack}
  \end{subfigure}

  \medskip

  \begin{subfigure}{.48\linewidth}
    \includegraphics[width=\linewidth]{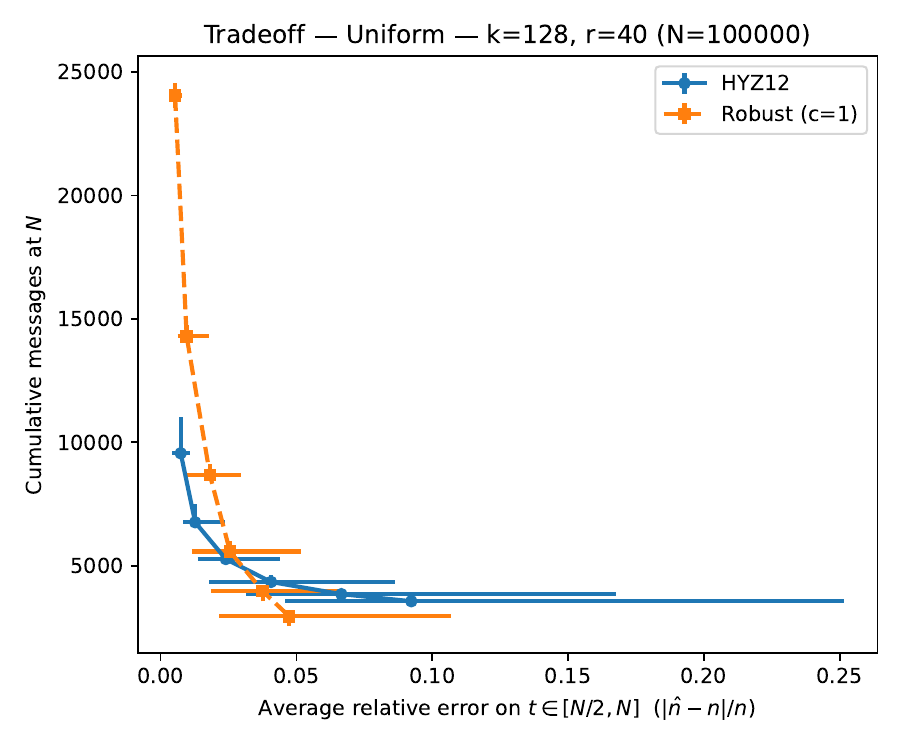}
    \caption{\(k=128\), Uniform}
  \end{subfigure}\hfill
  \begin{subfigure}{.48\linewidth}
    \includegraphics[width=\linewidth]{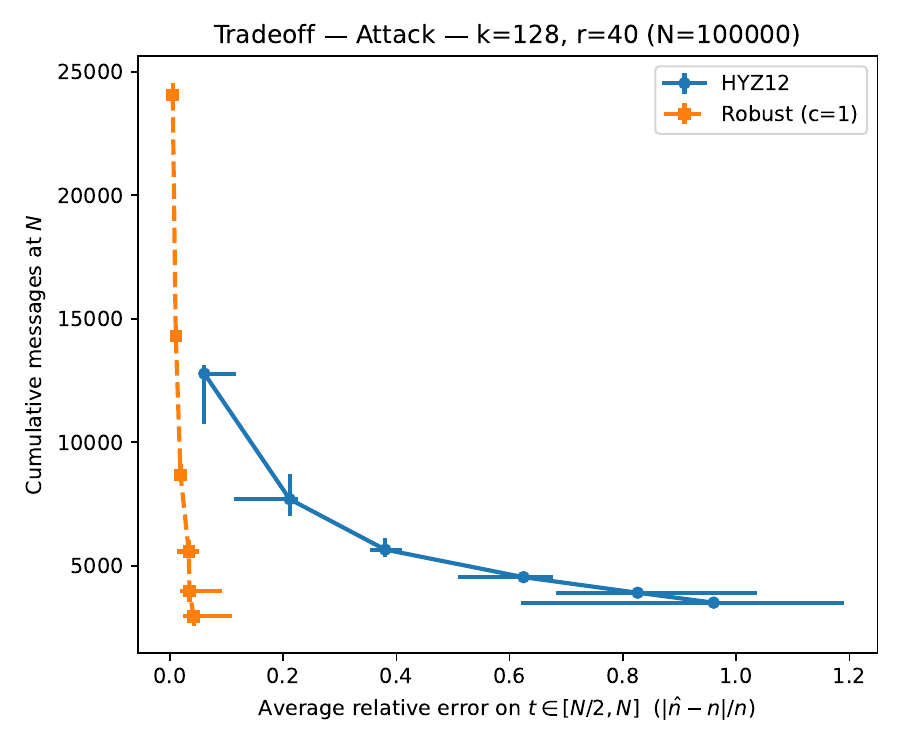}
    \caption{\(k=128\), Attack}
  \end{subfigure}
  \caption{Communication vs.\ accuracy tradeoff (median with (5\%,95\% bands) for \(N=10^5\), \(r=40\) (\textsf{Robust} uses \(c=1\)). Each panel shows the \(\varepsilon\) sweep; numeric \(\varepsilon\) values are omitted in the plot.}
  \label{fig:tradeoff-grid}
\end{figure}

\bibliographystyle{alpha}
\newcommand{\etalchar}[1]{$^{#1}$}

\appendix

\section{Tail bounds} \label{sec:tails}

\begin{claim}[Bound on the maximal partial-sum deviation {\cite[Ch.~2]{boucheron2013concentration}}]
Let $X_1,\ldots,X_r \stackrel{\mathrm{i.i.d.}}{\sim} \mathrm{Geom}(p)$ with support $\{1,2,\dots\}$, 
and define $S_i = \sum_{j=1}^i X_j$.  
Then for all $t > 0$,
\[
\Pr\!\left(\max_{1\le i\le r} \big|S_i - \tfrac{i}{p}\big| \ge t\right)
~\le~
2\exp\!\left(-\min\!\left\{\frac{t^2 p^2}{8r},\ \frac{t p}{4}\right\}\right).
\]
\end{claim}

\begin{proof}
Let $Y_i := X_i - \frac1p$ so that $\mathbb{E}[Y_i] = 0$ and 
$S_i - i/p = \sum_{j=1}^i Y_j$.  

\textbf{Step 1: Doob's maximal inequality.}  
For any $\lambda > 0$, the process
\[
M_i := \exp\!\left(\lambda \sum_{j=1}^i Y_j\right)
\]
is a non-negative submartingale.  
By Doob's maximal inequality (\cite[§2.5]{boucheron2013concentration}),
\[
\Pr\!\left(\max_{1 \le i \le r} \sum_{j=1}^i Y_j \ge t\right)
\ \le\ e^{-\lambda t} \, \mathbb{E} \, e^{\lambda \sum_{j=1}^r Y_j}
\ =\ \exp\!\big(-\lambda t + r \, \psi(\lambda)\big),
\]
where $\psi(\lambda) := \log \mathbb{E} e^{\lambda Y_1}$ is the cumulant generating function of $Y_1$.  
An identical bound with $-\lambda$ handles the lower tail, and a union bound gives a factor $2$.

\textbf{Step 2: Bernstein--type mgf bound.}  
For $X \sim \mathrm{Geom}(p)$ with support $\{1,2,\dots\}$,
\[
\mathbb{E} e^{\lambda X} = \frac{p \, e^\lambda}{1 - (1-p) e^\lambda}, 
\qquad \lambda < -\log(1-p).
\]
Thus, for $0 < \lambda < p$,
\[
\psi(\lambda) 
= \log \mathbb{E} e^{\lambda(X - 1/p)}
\le \frac{v \, \lambda^2}{2(1 - c \lambda)},
\]
where $v = \mathrm{Var}(X) = \frac{1-p}{p^2} \le \frac1{p^2}$ and $c = \frac1p$.  
This is exactly the sub-gamma mgf condition in \cite[Thm.~2.10]{boucheron2013concentration}.

\textbf{Step 3: Tail bound.}  
Plugging into Step~1 and optimizing over $\lambda$ via \cite[Cor.~2.11]{boucheron2013concentration} gives
\[
\Pr\!\left(\max_{1 \le i \le r} \sum_{j=1}^i Y_j \ge t\right)
\ \le\ \exp\!\left(-\frac{t^2}{2(rv + c t)}\right).
\]
Since $v \le 1/p^2$ and $c = 1/p$, we have
\[
\Pr\!\left(\max_{1 \le i \le r} (S_i - i/p) \ge t\right)
\ \le\ \exp\!\left(-\min\left\{\frac{t^2 p^2}{2r}, \ \frac{t p}{2}\right\}\right).
\]
Union bounding with the lower tail yields
\[
\Pr\!\left(\max_{1 \le i \le r} \big|S_i - i/p\big| \ge t\right)
\ \le\ 2 \exp\!\left(-\min\left\{\frac{t^2 p^2}{8r}, \ \frac{t p}{4}\right\}\right),
\]
after a slight relaxation of constants.
\end{proof}

\section{Lemmas for attack analysis} \label{sec:attacktoolsproofs}

This section contains deferred proofs for zero-inflated geometrics from \cref{sec:attacktools}.

\GeometricTelescopes*
\begin{proof}
Let $G_q(s)=\E[s^{\Geom(q)}]=\dfrac{q s}{1-(1-q)s}$ be the pgf of $\Geom(q)$.
For $q\le p$ the pgf of $Z_{q,p}$ is
\[
F_{q,p}(s)=\E[s^{Z_{q,p}}]
=\frac{q}{p}+\Bigl(1-\frac{q}{p}\Bigr)G_q(s)
=\frac{q}{p}\cdot\frac{1-(1-p)s}{1-(1-q)s}.
\]
Hence independence implies
\[
\prod_{i=1}^{r-1} F_{p_{i+1},p_i}(s)
=\prod_{i=1}^{r-1}\left(\frac{p_{i+1}}{p_i}\cdot
\frac{1-(1-p_i)s}{1-(1-p_{i+1})s}\right)
=\frac{p_r}{p_1}\cdot\frac{1-(1-p_1)s}{1-(1-p_r)s}
=F_{p_r,p_1}(s),
\]
where the middle equality is a telescoping product. Since pgfs coincide, the distributions coincide, proving the claim.
\end{proof}

\zeroinflatedtail*

\begin{proof}
For $\lambda\in(0,-\ln(1-p))$ the mgf of $\Geom(p)$ is
$M_G(\lambda)=\E[e^{\lambda G}]=\frac{pe^{\lambda}}{1-(1-p)e^{\lambda}}$.
By independence and $\ln(1+u)\le u$,
\[
\ln\E\!\left[e^{\lambda\left(X_i-\alpha_i/p\right)}\right]
=\ln\!\Bigl(1-\alpha_i+\alpha_i\,M_G(\lambda)e^{-\lambda/p}\Bigr)
\;\le\;\alpha_i\!\left(M_G(\lambda)e^{-\lambda/p}-1\right).
\]
Summing over $i$ and writing $S:=\sum_i X_i$ gives
\[
\ln\E\!\left[e^{\lambda(S-\mu)}\right]
\;\le\;A\!\left(M_G(\lambda)e^{-\lambda/p}-1\right).
\]
A direct calculus bound (using $-\ln(1-x)\le \tfrac{x}{1-x}$ and
$e^{\lambda}\le \tfrac{1}{1-\lambda}$ for $\lambda\in[0,1)$) yields, for
$0\le\lambda<p$,
\[
\ln\!\bigl(M_G(\lambda)e^{-\lambda/p}\bigr)
\;\le\;\frac{(1-p)}{p^2}\,\frac{\lambda^2}{1-\lambda/p}.
\]
Hence, for $0\le\lambda<p$,
\[
\ln\E\!\left[e^{\lambda(S-\mu)}\right]
\;\le\;\frac{A(1-p)}{p^2}\,\frac{\lambda^2}{1-\lambda/p}.
\]
By Chernoff's method,
$\Pr(S-\mu\ge t)\le \inf_{0<\lambda<p}\exp\!\bigl(-\lambda t
+ \frac{A(1-p)}{p^2}\frac{\lambda^2}{1-\lambda/p}\bigr)$.
Optimizing the quadratic-over-linear upper bound (standard Bernstein
optimization) gives
\[
\Pr(S-\mu\ge t)\ \le\
\exp\!\left(-\,\frac{t^2}{2\left(\dfrac{A(1-p)}{p^2}+\dfrac{t}{p}\right)}\right),
\]
and the min-form follows from $ \frac{u^2}{2(v+u)}\ge \frac12\min\{\tfrac{u^2}{v},u\}$.
\end{proof}

\end{document}